\documentclass[journal]{IEEEtran}
\usepackage{amsmath,amsfonts,amsthm}
\usepackage{array}
\usepackage[caption=false,font=footnotesize,labelfont=rm,textfont=rm]{subfig}
\usepackage{textcomp}
\usepackage{stfloats}
\usepackage{url}
\usepackage{verbatim}
\usepackage{graphicx}
\usepackage{cite}

\usepackage{amsmath,amsthm}
\usepackage{amssymb,boxedminipage}
\usepackage{amsfonts}
\usepackage{color,xcolor}
\usepackage{bm}

\usepackage{booktabs}
\usepackage{makecell}
\usepackage{ragged2e}
\usepackage[linkcolor=black,urlcolor=black,colorlinks,anchorcolor=black,citecolor=black]{hyperref}

\newtheorem{theorem}{Theorem}

\newtheorem{lemma}{Lemma}

\newtheorem{definition}{Definition}

\usepackage{float}
\usepackage{algorithmic}
\usepackage[linesnumbered,ruled,vlined]{algorithm2e}

\SetKwInput{KwInput}{Input}
\SetKwInput{KwOutput}{Output}

\begin{document}

\title{An Efficient Method for Realizing Contractions of Access Structures in Cloud Storage}

\author{Shuai Feng~and~Liang Feng Zhang
\thanks{S. Feng and L.F. Zhang (Corresponding author) are with the School of Information Science and Technology, ShanghaiTech University, Shanghai, PR China.
E-mail: \{fengshuai,zhanglf\}@shanghaitech.edu.cn}
}



\maketitle

\begin{abstract}
In single-cloud storage, ciphertext-policy attribute-based encryption (CP-ABE) allows one to encrypt any data under an access structure to a cloud server,
 specifying what attributes are required to decrypt. In multi-cloud storage, a secret sharing scheme (SSS) allows one to split any data into multiple shares, one to a single server,
 and specify which subset of the servers are able to recover the data. It is an
  interesting problem   to remove some attributes/servers but still enable the remaining attributes/servers in every authorized set to recover the data.
 The problem is related to the contraction problem  of access structures for SSSs.
In this paper, we propose a  method
 that can efficiently
 transform  a given SSS for an access structure  to
 SSSs  for contractions of  the access structure.
  We  show its applications in solving  the   attribute removal problem
  in the CP-ABE based single-cloud storage and the data relocating problem in
 multi-cloud storage.
Our method results in  solutions that  require either  less
server storage or even  no additional server storage.
\end{abstract}

\begin{IEEEkeywords}
Cloud storage, access structure, contraction, linear secret sharing, attribute-based encryption.
\end{IEEEkeywords}

\section{Introduction}

\IEEEPARstart{W}{ith} the rapid development of cloud computing in recent years, \emph{cloud storage} \cite{MTB+19,NCD+21} has moved to
 the mainstream of storage technology. It allows one to
 lease computing
 resources  from  cloud service providers  in a pay-per-use manner and
 remotely store/access important data, without
 need to build local data
 centers (including expensive software and hardware infrastructures) from scratch. However,
   many organizations are still reluctant to use cloud to store sensitive
    data. The reason is that they may lose  control of the  data,
   and  information leakage
     may occur due to unauthorized access \cite{DCT+18}.
     How to ensure the
     {\em confidentiality} of cloud data  is an important problem.
	
Storing data with cloud may use two different models:
{\em single-cloud storage} and
{\em multi-cloud storage}.
Single-cloud storage means storing data with a single cloud server.
In this model, the  confidentiality of data may be ensured by
the user encrypting the data and uploading the ciphertexts to
a cloud server. For example, \emph{ciphertext-policy
attribute-based encryption} (CP-ABE) \cite{BSW07,TJL+21} is
a commonly used   encryption technology for fine-grained access control that allows one to
 set a \emph{policy}
to specify who
are eligible to decrypt a ciphertext.
More precisely, every CP-ABE ciphertext is associated   with a policy; every
\emph{user} of the data gets a private key associated
  with several attributes from a set of $n$ {\em attributes} and is eligible to decrypt
 if and only if its attributes   satisfy the policy.
	
In multi-cloud storage \cite{FEJ15,ZYT+18,LQLL14,MTB+19}, one
may generate $n$ {\em shares} of the data   and store  each share  with a different cloud server
to enforce the following   policy: the   data can be reconstructed if
  $\geq t+1$ out of the $n$ shares are available but any $\leq t$ shares leak  information about the data.
In this model, special techniques for  splitting data into shares, such as
 homomorphic secret sharing
(HSS) \cite{BGI16},  have found interesting applications in the field of outsourcing
computations \cite{ZS19,ZW22}. Such techniques allow  each server to
compute a function on its share to produce a {\em partial result} and finally
a user of the data can
  reconstruct the function's
output from all   partial results.
	
	A critical technology used in both single-cloud storage
and multi-cloud storage
 is {\em linear secret sharing schemes} (LSSSs) \cite{OSZ19}.
    A  {\em secret sharing scheme} (SSS) \cite{BR79,Sha79} for a
  set ${\cal P}=\{P_1,\ldots,P_n\}$ of $n$  \emph{participants}   allows a \emph{dealer}
   to generate    $n$ \emph{shares} for a  \emph{secret} $s$, one for each participant, such that any \emph{authorized subset} of
$\cal P$ can reconstruct $s$ with their shares but  any \emph{unauthorized subset} learns   no information about $s$. The set $\Gamma$ of all authorized subsets is
called   an \emph{access structure} and the SSS is said to \emph{realize} $\Gamma$. In general, a set $\Gamma$ of subsets of $\cal P$ is qualified as   an access structure if the superset of any set in $\Gamma$ remains in $\Gamma$.
 SSSs were
 introduced  by Shamir \cite{Sha79} and Blakley \cite{BR79} for
 {\em threshold} access structures   and then extended to any
 {\em general} access structures  by  Ito et al. \cite{ISN89}.
 An SSS is \emph{linear} if both the {\em share generation}  and
the {\em secret reconstruction}  can be accomplished with linear operations.
 Since
 \cite{Sha79,BR79,ISN89},  SSSs have
  been one of the most important building blocks of cryptographic protocols   \cite{BGW88,CGKS95,SW05}.
In particular, the $n$ attributes in the
CP-ABE based single-cloud storage model and the $n$ servers in  multi-cloud storage model may   play the roles of
 the $n$ participants   in SSSs, respectively, and
 the  policies in both models may play the role of
 access structures.

 In this paper, we consider  application scenarios
  where part of the attributes/servers have to be removed such that the left
  attributes/servers in every
  authorized subset remain eligible to access data,  in order to make the access policies
  {\em  less restrictive}.
   Such scenarios may appear in both   storage models.
   For example, in the single-cloud storage model, a patient Alice may
   have encrypted her   electronic health records (EHRs)
 \cite{WYN+23,WZMZ23,YWW23} as a CP-ABE ciphertext with an access policy $\texttt{`hospital A'} \wedge
 \texttt{`branch 1'} \wedge \texttt{`respiratory'}$ such that Bob,
  a respiratory physician in the branch 1 of hospital A   can decrypt the ciphertext.
After an initial diagnosis,  Bob   may conclude that the condition  of Alice is
so complicated that  a consultation by the respiratory physicians
    from all branches (not just branch 1) of hospital A is needed. In this scenario,
    Alice needs to  update the ciphertext  and remove the attribute  \texttt{`branch 1'}
    from the policy such that all involved  physicians are able to decrypt.
    In the multi-cloud storage model,  organizations such as  transaction  platforms
    may collect tons of customer preference data and share the data among $n$ cloud servers. Users of the data may contact $t+1$ out of the $n$ servers, reconstruct the data, perform
    machine learning algorithms, and use the resulting model to   make higher profits.
The users need to pay for the services, as per  the total   amount of data downloaded from the $t+1$ servers.
    As the     data may lose relevance over time and damage the model's accuracy \cite{VHAI22}, both the value
and the privacy level of the data could be reasonably reduced over time.
It  is reasonable  for the organization to gradually reduce the  threshold
 $t$ such that less servers are needed to reconstruct the data over time.
In this scenario, the privacy threshold $t$ may be reduced by gradually closing some of the servers and {\em relocating} the data (shares) on
these servers to the remaining ones.

In both   application scenarios, it suffices for the data owners to consider the problem of
 {\em   how to remove
 an unauthorized subset of the  attributes/servers but still enable
 the remaining attributes/servers in every authorized subset to
 recover the  data}.
In the language of SSS, a more formal description  of the above problem  is as follows:
\begin{itemize}
\item [(p1)] \em
A secret $s$ has been shared according to an access structure $\Gamma$ over a set  ${\cal P}$ of participants and later an unauthorized subset $Q\subseteq {\cal P}$  have to be removed.
 How to  distribute   the shares of $Q$ to the participants
in
${\cal P} \setminus Q$ such that for every  authorized subset $A\in \Gamma$ the participants in
$A\setminus Q$
are still able to reconstruct $s$.
\end{itemize}
Via some abstraction, (p1) is related to the following problem:
\begin{itemize}
\item[(p2)] \em Given an SSS realizing an access structure $\Gamma$   and an unauthorized subset $Q$ of participants, construct a new SSS for $\Gamma_{\cdot Q}=\{A\setminus Q: A\in \Gamma\}$.
\end{itemize}
A solution to     (p2) may provide a solution to   (p1), if  given the shares of $Q$,  every  participant in ${\cal P}\setminus Q$ can combine with its own share to produce a new share, such that
 the new shares  realize
 $\Gamma_{\cdot Q}$ over ${\cal P}\setminus Q$ for the secret $s$.
  In the literature,  $\Gamma_{\cdot Q}$ has been called the \emph{contraction} of $\Gamma$ at $Q$ and the problem (p2) has been studied by \cite{Mar93,NN04}.
  In particular, the ideas of \cite{Mar93,NN04}
  can be extended to solve  (p1), either by
  $Q$ simply moving   their shares to a public storage or
 every other participant.
  However, both solutions  consume {\em additional storage}.
  In this paper, we are interested in solutions that require
  {\em  no additional storage}.

	\subsection{Theoretical Contributions}
	Informally, an SSS for $\Gamma$ is \emph{ideal} if all of the shares are from
	the same domain as the secret \cite{Bri89}. If there is an ideal SSS realizing $\Gamma$, then $\Gamma$ is ideal. In Section \ref{sec:our transformations}, we propose a solution for (p2) under ideal access structures. More precisely, we provide two algorithms: the first one is applicable to $|Q|=1$ and the second one is an extension of the first and is applicable to $|Q|>1$. It is well-known that the shares in every LSSS can be generated by a matrix. Our algorithms are efficient and apply simple linear transformations on the matrix that generates an SSS for $\Gamma$ to output a new SSS for $\Gamma_{\cdot Q}$. While for (p1), the transformation gives a method for $Q$ to distribute their shares: send the shares to every remaining participant and each remaining participant can apply the same transformations on its shares and the shares of $Q$ to obtain its shares in the new SSS. This will keep the size of each remaining participant's share unchanged. In particular, if we apply the proposed algorithm to an ideal LSSS for $\Gamma$, then an ideal LSSS for $\Gamma_{\cdot Q}$ will be obtained.

\subsection{Applications}
	\label{ss:apps}

Our algorithm have applications in both
 multi-cloud storage  and   CP-ABE based single-cloud storage.

\vspace{1mm}
	\noindent \textbf{Multi-cloud storage.}
 In multi-cloud storage, the data owners
split their valuable data into multiple shares and store shares on
multiple independent cloud servers such that the users who have paid for the services are
 eligible to access the data.
Many existing schemes for multi-cloud storage
\cite{WFZ+23,HK23,LCH22,SAR18,ZYT+18}   involve
 a considerable number of servers and the number may be as big as
 100. Closing   some of the servers is reasonable
as the data is gradually devaluing over time and the data owners need to
 economize expenses on server rental.
Our algorithms allow  the data owners to
 {\em  properly relocate} shares on some of the servers  to the other servers
 such that the data is recoverable by  downloading  shares from less servers.
In Section \ref{sec:appmc}, we show our solution and   compare   it  with
three existing solutions.
The comparisons show that our solution is most efficient in terms of cloud storage as
it  requires {\em no additional storage} on the remaining servers.
Our only price  is  a low cost  linear computation on the remaining servers.

\vspace{1mm}
 \noindent\textbf{Single-cloud storage.}
In the CP-ABE based single-cloud storage,
 some of the attributes may become unnecessary
 \cite{JSMG18} and the ciphertext has to be updated
 such that the decrypting information associated with these attributes
 is {\em properly associated} with  the remaining attributes, in order to
 change the access policy. A trivial method is
 downloading and decrypting the ciphertexts, and then re-encrypting
      the messages with $\Gamma_{\cdot Q}$. Its computation and communication
       costs may be high. In Section \ref{sec:app}, we
       propose CP-ABE-CAS,  a novel model of CP-ABE with contractions of access
       structures. In the proposed model, we introduce
        contraction keys, which are generated by the data owner itself, and a contraction
        algorithm such that: (1) the data owner is enabled to dominate the
         policy update, and (2) given the contraction key, the servers can efficiently update
          the            ciphertexts to adapt to a new access policy as per the data
           owner' preferences while the users' private keys remain unchanged.
We construct a   CP-ABE-CAS scheme based on
            Waters' CP-ABE scheme \cite{Wat11}. Experimental results show
            that our scheme may reduce the server-side storage cost and the
             user-side computation/communication cost through efficient update of ciphertexts on
             server-side.

	\subsection{Related Work}

\subsubsection{Attribute-Based Encryption in Single-Cloud Storage}

	\noindent \textbf{Attribute-based encryption.}
 Goyal  et al. \cite{GPSW06} classified ABE \cite{SW05} into two types: key-policy ABE (KP-ABE) and ciphertext-policy ABE (CP-ABE). In CP-ABE, a user's private key is associated with a set of attributes and a ciphertext is associated with a policy.

\vspace{1mm}
	\noindent \textbf{Policy updating.}
In CP-ABE, policy updating \cite{SSW12} refers to the problem
of  changing the access policy associated with a ciphertext. In \cite{SSW12}, the updated access policy is more restrictive than the original, so their construction {\em cannot} support the contraction studied by this work.     Ciphertext-policy attribute-based proxy re-encryption   \cite{LCLS09,YWRL10,LFSW13,LAS+14}
  allows a  semi-trusted proxy to perform policy updating.
Their scheme requires  a private key  whose associated attribute set satisfies the policy
to  generate the  re-encryption key
  and gives   the data owner  {\em no} control over the update of  access policy. In our work, no private key is   required and the data owner has {\em full} control
   over  the update of policy.

  \vspace{1mm}
	\noindent \textbf{Revocation.}
In  CP-ABE,  revocation  \cite{ASLK19} means   revoking the access
 right of a user such that the user is no loner able to decrypt
  a ciphertext that he  used to be able to decrypt.
  Revocation may happen when
  the services purchased by the users have expired.
  It is different from  contraction because
    contraction       removes some attributes from every authorized subset  so that more users  become eligible to decrypt.
For example, Ge et al. \cite{GSB+21} proposed a revocable
       ABE scheme with data integrity protection that
       adds  more attributes to every authorized set.
       The  new access structure is
		 more restrictive  and results in the revocation of
 some authorized users.

  \vspace{1mm}
  \noindent \textbf{Extendable access control.}
 An extendable access control system   \cite{SJG+17} allows a data owner to encrypt its data under an access structure $\Gamma$ and later allows any user whose attribute set satisfies $\Gamma$ to extend $\Gamma$ to a new access structure $\Gamma'$ such that
 any attribute set in   $\Gamma \cup \Gamma'$ is able to access the data.
 When  $\Gamma'=\Gamma_{\cdot Q}$ for some authorized subset $Q$, their scheme gives  contraction. Compared with us, the   data owner in \cite{SJG+17} has
 {\em no} control over the extended access structure and the resulting
   ciphertext becomes  \emph{longer} than the original one. In our work,   the data owner has {\em full} control on the contracted   access structure and
    the resulting ciphertext  is \emph{shorter} than the original one. Lai et al. \cite{LGS+21}  proposed a scheme in a different setting of
    \emph{identity-based} encryption.

	\subsubsection{Multi-Cloud Storage}
	The multi-cloud storage model \cite{FEJ15,LQLL14,MTB+19,XHL+20} has been very popular for
 ensuring data confidentiality in cloud environment. For example, Xiong et al. \cite{XHL+20}
 considered a problem of unbalanced  bandwidth between users and  servers in a
 multi-cloud storage system
and proposed {\em adaptive bandwidth} SSSs based on both Shamir's SSS \cite{Sha79} and Staircase codes.
In this work, we consider a different  problem of relocating data from some servers to
 the other servers. Our scheme can be used to improve the user's communication efficiency
by {\em closing} the servers with excessively low  bandwidth.
	
	\subsubsection{Secret Sharing}
 \

	\noindent \textbf{Dynamic secret sharing.} SSSs that allow dynamic changes to
 access structures have been designed in \cite{Cac95,YYQ09,YL19,TKH15} and called \emph{dynamic} SSSs.
	For  general access structures, the scheme proposed by Cachin \cite{Cac95} allows the participants
 to be added or removed. In particular, when the participants in some
 unauthorized
  subset are removed, their shares will be published so
  that the scheme can realize contraction of the access structure at the subset and \emph{occupy additional storage}.
   Our solution requires \emph{no additional storage}. The schemes in \cite{YYQ09,YL19} did \emph{not}
    consider contractions of access structures when removing some participants.
    In fact, the contraction of an access structure $\Gamma$ at a set $Q$ can also be the union of $\Gamma_{\cdot Q}$ and $\Gamma$, so the access structure
     can also contract through adding {$A\setminus Q$} (where $A\in \Gamma$) as new authorized subsets.
     The schemes in \cite{YYQ09,YL19} allow one to add new authorized subsets  {and are  {\em computationally secure}.
      Our work uses {\em information-theoretic} SSSs.} The schemes of \cite{TKH15} can only realize  \emph{threshold} access structures rather than general access structures. Our transformation is applicable to \emph{any} LSSS, even if the access structure is not threshold.
	
\vspace{1mm}
	\noindent \textbf{Proactive secret sharing.} There is a long line of research on SSS that enables participants
  to update their shares such that the information obtained by any adversary will be invalid.
   Such an SSS was introduced by Herzberg et al. \cite{HJKY95} and called a
   \emph{proactive} SSS. The schemes proposed in
    \cite{HJKY95,Mas17} are only applicable to \emph{static committees}.
      Later, \emph{dynamic} proactive SSSs have been proposed
      in \cite{SLL10,MZW19,NNPV02,NNPV02a}, both for 
              \emph{threshold} access structures \cite{SLL10,MZW19}
              and for general access        structures  \cite{NNPV02,NNPV02a}. But if we focus on
       contractions of access structures, in \cite{NNPV02,NNPV02a}, \emph{the remaining participants are required
        to interact with each other}.
        In contrast, there is \emph{no interaction among the remaining
         participants} in our work.
	
\vspace{1mm}
	\noindent \textbf{Contraction.} Closest to our work are \cite{Sli20,Mar93,NN04}.
Slinko \cite{Sli20} studied several ways to merge two ideal linear SSSs into a new ideal
 linear SSS but did \emph{not} consider contractions of access structures. The work in
  \cite{Mar93,NN04} solved the problem (p2). For the contraction of an access structure
  $\Gamma$ at a set $Q$ of the removed participants, if we assume that the share size of
  each participant is $\ell$, the solutions in
  \cite{Mar93,NN04} require the contracted system to consume \emph{additional storage of at least $|Q|\cdot\ell$}
  to store the shares of the removed participants. In our work, the proposed construction
   allows the remaining participants to combine the shares of $Q$ with their shares to produce new shares so that \emph{no additional storage} is needed.

	\subsection{Organization}
	 Section \ref{sec:preliminaries} provides some basic  definitions and notations.  {Section \ref{sec:our transformations} solves the problem (p2).}
	 Section \ref{sec:appmc} and Section \ref{sec:app} solve the problem (p1) in multi-cloud storage and single-cloud storage, respectively. Finally,
  Section \ref{sec:concluding remarks} contains our concluding remarks.

\section{Preliminaries}
\label{sec:preliminaries}

	For any integer $n>0$, we denote $[n]=\{1,\ldots,n\}$.
	For any vector $\bm s=(s_1,\ldots,s_n)$ and any set $I=\{i_1,\ldots,i_k\}\subseteq [n]$, we denote   $\bm s_I=(s_{i_1},\ldots,s_{i_k})$. In particular, we will write $\bm s_I=(s_i)_{i\in I}$.
	For any $d$-dimensional  vector $\bm t$, we denote  $t_j=\bm t_{\{j\}}$ for any $j\in [d]$.
	For any finite set $\mathcal P$, we   denote by $2^{\mathcal P}$ the \emph{power set} of
	$\mathcal P$, i.e., the set of all subsets of $\mathcal P$.
	Let $\psi:A\rightarrow B$ be a function. For any $b\in B$ (resp. any $C\subseteq B$), we denote $\psi^{-1}(b)=\{a\in A: \psi(a)=b\}$ (resp.  $\psi^{-1}(C)=\{a\in A: \psi(a)\in C\}$).

	\begin{definition}[{\bf Access Structure} \cite{Bei11}]
		Let $\mathcal P=\{P_1,\ldots,P_n\}$ be a set of $n$ participants. A collection $\Gamma\subseteq 2^{\mathcal P}$ is said to be
		\emph{monotone} if it satisfies the property: For any $A,B\in 2^{\mathcal P}$, if $A\in \Gamma$ and $A\subseteq B$, then $B\in \Gamma$. \noindent A collection $\Gamma\subseteq 2^{\mathcal P}$ is said to be an \emph{access structure} over  $\mathcal P$ if it consists of non-empty subsets of $\mathcal P$   and is monotone.
	\end{definition}
	
	Let $\Gamma$ be an access structure over   $\mathcal P=\{P_1,\ldots,P_n\}$.
	Every set in   $\Gamma$ is said to be an \emph{authorized} subset of $\mathcal P$.
	Every set in $2^{\mathcal P}\setminus\Gamma$ is said to be \emph{unauthorized}.  An authorized subset $A\in \Gamma$ is \emph{minimal} in $\Gamma$ if no proper subset of $A$ belongs to $\Gamma$.
	The \emph{basis} of $\Gamma$ consists of   all minimal authorized subsets in $\Gamma$ and  denoted as $\Gamma^-$. 	
	The access structure $\Gamma$ is said to be  \emph{connected} if each participant
	$P_i\in \mathcal P$ belongs to at least one minimal authorized subset in $\Gamma^-$.
	
	Any  access structure can be realized by a secret sharing scheme, which is essentially a distribution scheme with privacy properties.
	\begin{definition}[{\bf Distribution Scheme} \cite{Bei11}]
		\label{DS}
		Let $\mathcal P=\{P_1,\ldots,P_n\}$ be a set of $n$ participants and let
		$S$ be the domain of secrets.
		A \emph{distribution scheme} for sharing the secrets
		in $S$ among the participants in $\mathcal P$ is a pair  $\Pi=(\pi, \mu)$, where $\mu$ is a probability distribution over a finite set
			$R$ of \emph{random strings}, and $\pi:S\times R \rightarrow S_1\times\cdots\times S_n$ is a mapping ($S_i$ is the domain of \emph{shares} of $P_i$ for every $i\in[n]$).
	\end{definition}
	
	With the scheme $\Pi$, a dealer can distribute a secret $s\in S$ by firstly choosing a random string $r\in R$ according to $\mu$, computing a vector of shares $\pi(s,r)=(s_1,\ldots,s_n)$, and privately communicating each
		share $s_i$ to   $P_i$. For a set $A\subseteq \mathcal P$, we denote by $\pi_A(s,r)=(s_i)_{i\in I_A}$ the restriction of $\pi(s,r)$ to $I_A=\{i\in[n]: P_i\in A\}$.
		The efficiency  of $\Pi$  can be  measured by its \emph{information rate}
		$\rho(\Pi)=\log{|S|}/{\max_{i\in [n]}\log{|S_i|}}$.
	
	Without loss of generality, we  can always assume that   $\mu$ is the uniform distribution over a
	properly chosen set $R$ of random strings.
	When $\mu$ is the uniform distribution over $R$, we shall   denote
	$\Pi=\pi$, instead of $\Pi=(\pi,\mu)$.

	\begin{definition}[{\bf Secret Sharing Scheme (SSS)} \cite{Bei11}]
		Let $\mathcal P=\{P_1,\ldots,P_n\}$ be a set of $n$ participants.
		Let $\Gamma$ be an access structure over $\mathcal P$.
		Let $\pi: S \times R\rightarrow S_1\times \cdots \times S_n$  be
		a  distribution scheme for $\mathcal P$.
		The scheme $\pi$ is said to be a \emph{secret sharing scheme} realizing  $\Gamma$  if the following  requirements are satisfied:

		\noindent{\bf Correctness.}  Any  authorized subset of participants can reconstruct a secret  with their shares of the secret. Formally, for every  authorized subset $A= \{P_{i_1},\ldots,P_{i_m}\}\in \Gamma$, there is a reconstruction function $\mathsf{Recon}_A: S_{i_1}\times  \cdots\times S_{i_{m}}\rightarrow S$ such that for any  $s\in S$,
			$\Pr_{r}[\mathsf{Recon}_A\left(\pi_A\left(s,r\right)\right)=s]=1$.
		
\noindent{\bf Perfect Privacy.} Any unauthorized subset of participants
			cannot learn any  information about a secret from their shares of the secret. Formally, for every  unauthorized subset $A=\{P_{i_1},\ldots,P_{i_m}\} \in 2^{\mathcal P}\setminus \Gamma $, for any  $a,b\in S$ and any $\bm s=(s_{i_1}, \ldots,s_{i_m})\in
			S_{i_1}\times \cdots \times S_{i_m}$, $\Pr_r[\pi_A(a,r)=\bm s] =\Pr_r[\pi_A(b,r)=\bm s]$.
	\end{definition}
	
	Beimel \cite{Bei11} showed that for any SSS
	$\pi: S \times R\rightarrow S_1\times \cdots \times S_n$ realizing a connected access structure $\Gamma$, it must be that
	$|S_i|\ge |S|$ for all $i\in[n]$. Thus, the information rate of an SSS for a connected access structure is always $\le 1$. An SSS with information rate   1 is said to be \emph{ideal}. An access structure $\Gamma$ is  \emph{ideal} if there is an ideal SSS realizing $\Gamma$.

	\begin{definition}[{\bf Linear Secret Sharing Scheme (LSSS)} \cite{Wat11}]
		\label{LSSS}
		Let $\mathbb F$ be a finite field. Let $\mathcal P=\{P_1,\ldots,P_n\}$ be a set of $n$ participants.
		Let $\Gamma$ be an access structure over $\mathcal P$.
		Let $\pi: S \times R\rightarrow S_1\times \cdots \times S_n$ be an  SSS realizing $\Gamma$.
		The scheme $\pi$ is said to be \emph{linear} over $\mathbb F$ if
		there exist  a matrix $\bm H=(\bm h_1,\ldots,\bm h_\ell)^\top\in \mathbb{F}^{\ell\times d}$, a target vector $\bm t=(1,0,\ldots,0)^\top\in \mathbb{F}^d$, and a surjective function $\psi:[\ell]\rightarrow\mathcal P$ such that
		the share generation and secret reconstruction procedures are done as follows:
		
		\noindent{\bf Share Generation.} To share a secret $s\in \mathbb{F}$, $d-1$ random field elements $r_2, \ldots,r_d\in\mathbb{F}$ are chosen to form a vector $\bm v=(s,r_2,\ldots,r_d)^\top$. For   every $i\in[n]$,  the  participant $P_i$'s   share $\bm s_i$ is computed as
			$\bm s_i=(\bm h_j^\top \bm v)_{j\in \psi^{-1}(P_i)}$.

		 \noindent{\bf Reconstruction.} For any authorized subset $A\in\Gamma$,  there exist constants $\{\alpha_j:\psi(j)\in A\}$ such that
			$\bm{t}=\sum_{j \in \psi^{-1}(A) }\alpha_j \bm h_j$,
			and thus
$s=\bm t^\top\bm v=\sum_{j\in \psi^{-1}(A) }\alpha_j (\bm{h}_j^\top\bm v).$
	\end{definition}
	
	In Definition \ref{LSSS}, the tuple $\mathcal M=(\mathbb F,\bm H,\bm t,\psi)$ specifies an LSSS
	for $\Gamma$ and
	has been called a \emph{monotone span program} (MSP) \cite{KW93} for $\Gamma$.
	Beimel \cite{Bei11} showed that LSSSs  and MSPs are equivalent: {\em every LSSS for $\Gamma$ can be derived from an MSP for
	$\Gamma$ and vice versa.}
	
	In this paper, we are interested in the transformation from
	an LSSS for $\Gamma$
	to new LSSSs for contractions of $\Gamma$.
	\begin{definition}[{\bf Contraction} \cite{Mar93}]
		\label{contraction}
Let $\mathcal P=\{P_1,\ldots,P_n\}$ be a set of $n$ participants.
		Let $\Gamma$ be an access structure over $\mathcal P$.
		For any $Q\subseteq \mathcal P$, the \emph{contraction} of $\Gamma$ at $Q$, denoted as $\Gamma_{\cdot Q}$, is an access structure   on $\mathcal P\setminus Q$ such that for every $A\subseteq \mathcal P\setminus Q$, $A\in \Gamma_{\cdot Q} \Leftrightarrow A\cup Q\in \Gamma$.
We also say that  $\Gamma$ is \emph{contracted} at $Q$ to $\Gamma_{\cdot Q}$.
	\end{definition}

If $Q\in \Gamma$,  then $(\Gamma_{\cdot Q})^{-}=\{\{P_i\}|P_i\in\mathcal P\setminus Q\}$. If $Q\in 2^\mathcal P\setminus\Gamma$, then $(\Gamma_{\cdot Q})^-$ consists of all the minimal  nonempty subsets of the form $A\cap(\mathcal P\setminus Q)$, where $A$ is taken over $\Gamma^-$. For any two disjoint subsets $Q_1, Q_2\subseteq \mathcal P$,  $(\Gamma_{\cdot Q_1})_{\cdot Q_2}=\Gamma_{\cdot (Q_1\cup Q_2)}$.

	\section{Our Transformations}
	\label{sec:our transformations}
	
	In this section, we show how to transform an ideal LSSS for
	an access structure $\Gamma$ to ideal LSSSs for contractions of $\Gamma$.
	
	Let $\pi: S \times R\rightarrow S_1\times S_2\times \cdots \times S_n$ be an ideal
	LSSS for a {\em connected} access structure
	$\Gamma$ on $\mathcal P=\{P_1,\ldots,P_n\}$. {
 W}e suppose that $\pi$ is equivalent to an
	MSP
$\mathcal M=(\mathbb F, \bm H,\bm t,\psi)$, where $\mathbb{F}$ is a finite field, $\bm H=(\bm h_1,\ldots, \bm h_\ell)^\top$ is an $\ell\times d$ matrix over $\mathbb{F}$,
	$\bm t=(1,0,\ldots,0)^\top\in \mathbb{F}^d$ is a target vector, and $\psi:[\ell]\rightarrow \mathcal P$ is a surjective function that assigns the $\ell$ rows of $\bm H$ to the $n$ participants in $\mathcal P$.
	Since $\pi$ is ideal, we must have that  $|S_i|=|S|$ for every $i\in[n]$, $\ell=n$ and $\psi:[n]\rightarrow \mathcal P$ is a bijection. Without loss of generality, we can suppose that $\psi(i)=P_i$ for every $i\in[n]$.

The following lemma shows that for any unauthorized subset of participants,
if the rows assigned to them form a submatrix of $\bm H$  of rank $r$, then
the last $d-1$ columns of the submatrix must contain an invertible
submatrix of order $r$.
This $r\times r$ submatrix will be used in our transformations.
\begin{lemma}
\label{lem:U}
Let  ${\mathcal M}=(\mathbb F, \bm H,\bm t,\psi)$ be  an MSP that realizes   a connected access structure $\Gamma$ over $\mathcal P$. Let  $Q$ be any unauthorized subset and let ${\bm H_Q}=\big((\bm h_i)_{\psi(i)\in Q}\big)^\top$.
If ${\rm rank}({\bm H}_Q)=r$, then
 there exists a set $W=\{w_1,\ldots,w_r\}\subseteq \psi^{-1}(Q)$ and a set  $K=\{k_1,\ldots,k_r\}\subseteq [d]\setminus\{1\}$ such that the order-$r$ square matrix $\bm U=((\bm h_{w_1})_K,\ldots,(\bm h_{w_r})_K)^\top$ is invertible over $\mathbb{F}$.
 {\em (see Appendix \ref{app:lem1} for the proof)}
 \end{lemma}

	We will start with an algorithm ({Algorithm \ref{alg:s1}}) that
	takes    $\mathcal M$
	and an unauthorized subset $Q \subseteq \mathcal P$ with $|Q|=1$ as input (w.l.o.g., $Q=\{P_n\}$) and
	outputs a new ideal LSSS
	$\mathcal M'=(\mathbb F, \bm H',\bm t,\psi')$ for  $\Gamma_{\cdot Q}$, and then show an extended algorithm  ({Algorithm \ref{alg:s2}}) for any unauthorized subset
 $Q$  with $|Q|=m>1$ (w.l.o.g.,  $Q=\{P_{n-m+1},\ldots,P_n\}$).
	\begin{algorithm}
		\label{alg:s1}
		\caption{Contraction at $Q$ with $|Q|=1$}
		\KwInput{$\mathcal M=(\mathbb F, \bm H,\bm t,\psi)$, $Q=\{P_n\}$}
		\KwOutput{$\mathcal M'=(\mathbb F, \bm H',\bm t,\psi')$}
		
		 Choose $k\in [d]\setminus \{1\}$ such that $h_{nk}\ne 0$;
		
		\For{$i\in [n-1]$}{
			$\bm{h}_i'=\bm{h}_i-\frac{h_{ik}}{h_{nk}}\bm{h}_n$;
		}
		
		$\bm{H}'=(\bm{h}_1',\ldots,\bm{h}_{n-1}')^\top$;
		
		Define $\psi':[n-1]\rightarrow \mathcal P\setminus Q$  such that
		$\psi'(i)=P_i$ for every $i\in [n-1]$;
		
		\Return $\mathcal M'=(\mathbb F,\bm H',\bm t,\psi')$;

	\end{algorithm}

The step 1 of  Algorithm \ref{alg:s1}
{is always feasible}, due to Lemma \ref{lem:U}.
 It is also clear that the output  $\mathcal M'$ of {Algorithm \ref{alg:s1}}  gives an ideal SSS. Below we show that $\mathcal M'$ realizes $\Gamma_{\cdot Q}$.
	
	\begin{theorem}
		\label{1}
		If $\mathcal M$ is an ideal LSSS realizing the access structure
		$\Gamma$ over $\mathcal P=\{P_1,\ldots,P_n\}$, then for $Q=\{P_n\}\in 2^{\mathcal P}\setminus \Gamma$,  the ideal LSSS $\mathcal M'$ output by   Algorithm    \ref{alg:s1} realizes $\Gamma_{\cdot Q}$. {\em (see Appendix \ref{app:thm1} for the proof)}
	\end{theorem}

	For   $Q=\{P_{n-m+1},\ldots,P_n\}$,
	we can  iteratively
	performing  {Algorithm \ref{alg:s1}}
	for all $\{P_j\}\subseteq Q$.
	However, there is a simpler  one-step transformation (see  {Algorithm \ref{alg:s2}}).
	
	\begin{algorithm}
		\label{alg:s2}
		\caption{Contraction at $Q$ with $|Q|>1$}
		\KwInput{$\mathcal M=(\mathbb F, \bm H,\bm t,\psi)$, $Q=\{P_{n-m+1},\ldots,P_n\}$}
		\KwOutput{$\mathcal M'=(\mathbb F, \bm H',\bm t,\psi')$}
		
		Compute   {$r={\rm rank}((\bm h_{n-m+1},\ldots, \bm h_n)^\top)$};
		
 Find  a set   $W=\{w_1,\ldots,w_r\}\subseteq [n]\setminus [n-m]$
		and a set $ K=\{k_1,\ldots,k_r\}\subseteq [d]\setminus\{1\} $ such that  the   square  matrix $\bm{U}=((\bm h_{w_1})_K,\ldots,(\bm h_{w_r})_K)^\top$ is invertible;
		
		\For{$i\in [n-m]$}{
			Compute a new vector $\bm h_i'$ such that
$\bm{h}_i'^\top=\bm{h}_i^\top-(\bm h_{i}^\top)_K\cdot \bm{U}^{-1}\cdot (\bm{h}_{w_1}, \ldots, \bm{h}_{w_r})^\top$;
		}
		
		$\bm{H}'=(\bm{h}_1',\ldots,\bm{h}_{n-m}')^\top$;
		
		Define
		$\psi': [n-m]\rightarrow \mathcal P\setminus Q$ such that $\psi'(i)=P_i$
		for every $i\in[n-m]$;
		
		\Return $\mathcal M'=(\mathbb F,\bm H',\bm t,\psi')$;
		
	\end{algorithm}
	
	Likewise, the step 2 of   {Algorithm \ref{alg:s2}} {is always feasible} due to
Lemma \ref{lem:U} and  the output   $\mathcal M'$ of Algorithm \ref{alg:s2} is ideal. Theorem \ref{Qmulti} shows that $\mathcal M'$ exactly realizes $\Gamma_{\cdot Q}$.

	\begin{theorem}
		\label{Qmulti}
		If $\mathcal M$ is an ideal LSSS realizing the access structure
		$\Gamma$ over $\mathcal P=\{P_1,\ldots,P_n\}$, then for $Q=\{P_{n-m+1},\ldots,P_{n}\}\in 2^\mathcal P\setminus\Gamma$, the ideal LSSS $\mathcal M'$ output by  {Algorithm \ref{alg:s2}}  realizes  $\Gamma_{\cdot Q}$.
{\em (see Appendix \ref{app:thm2} for the proof)}
	\end{theorem}

\section{Application in Multi-cloud Storage}
	\label{sec:appmc}

	In this section, we show an application of our transformation in multi-cloud storage.
As stated in problem (p1), the scenario we will consider is as follows:
A user has used an ideal LSSS $\pi$ for an ideal access structure $\Gamma$ to share a
	secret $s$ as $n$ shares $s_1,\ldots,s_n$ and stored the share $s_i$ with a server $P_i$ for every $i\in[n]$. Later the user may want to   unsubscribe the servers in some unauthorized subset
	$Q\subseteq \{P_1,\ldots,P_n\}$.
We need to figure out how to   distribute the shares of $Q$ to the servers in $\mathcal P\setminus Q$ such that  for every  authorized subset $A\in \Gamma$ the servers in
$A\setminus Q$
are still able to reconstruct $s$.
In particular,  the dealer should   be not involved in the process.
In this section, we will discuss four possible solutions, the first one of which is
based on our transformation from Section \ref{sec:our transformations}, and show that ours
is superior than the others.
Furthermore, to understand clearly the four solutions, a toy example is given in Appendix \ref{app_toyexmaple}.

	\subsection{Solutions to  the  Storage Relocation Problem}
 \label{sec:mcs}
	\subsubsection{Our Method}\label{ourmd}
	Referring to Algorithm \ref{alg:s2}, let
  $\mathcal M$ be the ideal LSSS $\pi$ for $\Gamma$ and let $Q=\{P_{n-m+1},\ldots,P_n\}$.
On  input $(\mathcal M,Q)$, Algorithm \ref{alg:s2}  provides a solution to (p2) and   outputs
an ideal LSSS $\mathcal M'=(\mathbb F,\bm H',\bm t,\psi')$ for $\Gamma_{\cdot Q}$, where
${\bm H}'=({\bm h}'_1,\ldots,{\bm h}'_{n-m})^\top$ and  every  ${\bm h}'_i$
is computed as
\begin{align}
\label{eqn:hpi}
\bm{h}_i'^\top=\bm{h}_i^\top-(\bm h_{i}^\top)_K\cdot \bm{U}^{-1}\cdot (\bm{h}_{w_1}, \ldots, \bm{h}_{w_r})^\top
\end{align}
at step 4 of Algorithm \ref{alg:s2}.
If ${\bm v}=(s,r_2,\ldots,r_d)$ is  used for computing the original
 shares $s_1,\ldots,s_n$ of all servers, i.e.,
 $s_i={\bm h}^\top_i\cdot {\bm v}$ for all $i\in[n]$, then
$\{s'_i=({\bm h}'_i)^\top\cdot{\bm v}: i\in [n-m]\}$
will be $n-m$ shares that realize the scheme ${\cal M}'$ for sharing $s$.
In particular, as per  (\ref{eqn:hpi}), we have that
\begin{align}
\label{eqn:spi}
s'_i=s_i-(\bm h_{i}^\top)_K\cdot \bm{U}^{-1} (s_{w_1},\ldots,s_{w_r})^\top
\end{align}
is a linear combination of $P_i$'s share $s_i$ and $Q$'s shares
$s_{w_1},\ldots,s_{w_r}$ with constant  coefficients.
Our method for (p1) simply {\em requires each server $P_i$ to
perform the computation of {\em (\ref{eqn:spi})}
and store the new share $s'_i$.}

 Intuitively, our method for (p1) requires every remaining server to {\em linearly combine} its share with the shares of the removed servers.  Thereby for any LSSS for access structure $\Gamma$ and any unauthorized subset $Q$, we represent the new scheme for $\Gamma_{\cdot Q}$ obtained by using this method as $\pi_{\rm lc}$.

 In particular, a more intuitive solution for (p1) has been shown in Appendix \ref{appendixE} and is {equivalent} to our method. So its performance is the same as our method.

	\subsubsection{Martin's Method}\label{MT}

	In Martin \cite{Mar93}, an SSS $\pi: S \times R\rightarrow S_1\times S_2\times \cdots \times S_n$  for $\Gamma$ is represented as a matrix $\bm{M}$ that has  $|S\times R|$ rows
	and $n$  columns. Each row of the matrix is labeled with a pair
	$(s,r)\in S\times R$ and for every $i\in[n]$, the $i$th column of the matrix
	is labeled with $P_i$. For any $(s,r)\in S\times R$,
	and any $i\in[n]$,
	the entry of $\bm M$ in  row $(s,r)$ and   column
	$P_i$ is defined as the $i$th element of $\pi(s,r)$, i.e.,
	  $P_i$'s share of  $s\in S$
	when $r\in R$ is used as the random string for sharing.
	Given $\bm M$ and $Q$,
	Martin \cite{Mar93} has a transformation from $\bm M$ to a  SSS $\bm M'$
for $\Gamma_{\cdot Q}$. For 	$Q=\{P_{n-m+1},\ldots,P_{n}\}$,
the transformation can be described as   follows:
		\begin{itemize}
\em
			\item Choose    $\bm \alpha=(\alpha_1,\ldots,\alpha_m)\in S_{n-m+1}\times \cdots\times S_{n}$;
			\item  Define $\bm M'$  as the submatrix of $\bm M$ with rows
 labeled by $\{(s,r): (s,r)\in S\times R,      \pi_Q(s,r)=\bm \alpha\}$ and columns labeled by $\mathcal P\setminus Q$.
		\end{itemize}

	The transformation as above provides a   solution to   (p2).
To  extend it to solve (p1), the key point is enabling  the servers in $\mathcal P\setminus Q$ to have access to the shares of $Q$. A possible method  is to
{\em transfer $Q$'s shares
$\bm\alpha$
 to a pubic storage} so that the servers in $\mathcal P\setminus Q$ can determine
    $\bm M'$. The new scheme from Martin's method remains \emph{ideal} and is represented as $\pi_{\rm ps}$.

	\subsubsection{Nikov-Nikova Method}\label{NN}

Let $\mathcal M=(\mathbb F,\bm H,\bm t,\psi)$ be an MSP and an ideal LSSS for
	$\Gamma$. Let $Q$ be an unauthorized subset.
Nikov and Nikova \cite{NN04} has a method of  constructing an MSP $\mathcal M'=(\mathbb F,\bm H',\bm t,\psi')$ for $\Gamma_{\cdot Q}$ out of
 $(\mathcal M,Q)$:
		\begin{itemize}
\em
			\item  Let $\bm H_{\psi^{-1}(Q)}$ be the rows assigned to $Q$ in   $\mathcal M$. The matrix $\bm H'$ is obtained by appending $n-m-1$ copies
			of $\bm H_{\psi^{-1}(Q)}$ at the end of $\bm H$.
			\item The  map $\psi'$ is defined such that each server $P_i\in \mathcal P\setminus Q$ is assigned both one of the $n-m$ copies of $\bm H_{\psi^{-1}(Q)}$
			and the rows  $\bm H_{\psi^{-1}(P_i)}$.
		\end{itemize}

 More precisely, this solution to (p2) requires each server  in $Q$
	to make its share  accessible to  $\mathcal P\setminus Q$.
For (p1), this idea simply {\em requires  the servers in $Q$ to transfer their shares to every server
in $\mathcal P\setminus Q$}. Every server in $\mathcal P\setminus Q$ needs additional storage to {\em individually store a copy of $Q$'s shares}. The new scheme ${\cal M}'$ turns out to be \emph{non-ideal} and is represented as $\pi_{\rm is}$.
	
	\subsubsection{Extended Nikov-Nikova Method}\label{ENN}

	In Nikov-Nikova method, every server in $\mathcal P\setminus Q$ has to store a copy of $Q$'s shares. When $Q$ is removed, for every   $A\in \Gamma$, to enable the servers in $A\setminus Q$ to reconstruct $s$, $A\setminus Q$ only need to
{\em collectively store a copy of the shares of $Q$}. Thus the shares of $Q$ may be properly hand out to $\mathcal P\setminus Q$ such that every server in $\mathcal P\setminus Q$ only needs to hold a part of   $Q$'s shares, and for every $A\in \Gamma$,   $A \setminus Q$
 can reassemble the shares of $Q$ and then combine with their old shares to recover $s$. The new SSS is \emph{non-ideal}. We represent the new scheme as $\pi_{\rm cs}$.

	\subsection{Comparison}
	\label{cstas}
	In this section, we restrict our attention to the threshold access structures and compare among four methods mentioned in Section \ref{sec:mcs}, in terms of storage and information rate of the new scheme after contraction.
	
	\subsubsection{Theoretical Analysis}
	Let $t,n$ be integers such that $1\leq t\leq n$ and let $\mathcal P=\{P_1,\ldots,P_n\}$
	be a set of $n$ servers.
	Shamir's $(t,n)$-threshold secret sharing scheme (TSSS) \cite{Sha79} realizes a $t$-out-of-$n$ threshold access structure $\Gamma=\{A\subseteq \mathcal P:|A|\ge t\}$.
	In Shamir's scheme $\pi_0$,    a finite field $\mathbb{F}_p$ of prime order $p>n$ is chosen as the domain of secrets and the $n$ servers $P_1,\ldots,P_n$ are associated with
	$n$ distinct nonzero  field elements  $x_1,\ldots,x_n\in \mathbb F_p$, respectively. To share a secret $s\in \mathbb F_p$, the dealer chooses $t-1$  field elements $a_1,\ldots,a_{t-1}\in \mathbb F_p$ randomly, defines a  polynomial $P(x)=s+\sum_{j=1}^{t-1}a_jx^j$, and assigns a share $s_i=P(x_i)$ to the  server $P_i$ for all $i\in[n]$. Any $\ge t$ servers can
	reconstruct $s$ by  interpolating the polynomial $P(x)$ with their shares. It has been well showed in \cite{Sti92} that Shamir's $(t,n)$-TSSS is ideal and linear.
	That is, for Shamir's $(t,n)$-TSSS $\pi_0$, there exists an MSP $\mathcal M=(\mathbb F_p,\bm H,\bm t,\psi)$ where $\bm H=(\bm h_1,\ldots,\bm h_n)^\top\in \mathbb F^{n\times t}$ such that $\bm h_i=(x_i^0,\ldots,x_i^{t-1})$, $\bm t=(1,0,\ldots,0)^\top \in \mathbb F_p^t$, and $\psi(i)=P_i$ for each $i\in [n]$.
	In particular, for any unauthorized subset $Q$ of size $m$, $\Gamma_{\cdot Q}$ will be a $(t-m)$-out-of-$(n-m)$ threshold access structure. We shall apply our method (Section \ref{ourmd}), Martin's method (Section \ref{MT}), Nikov-Nikova method (Section \ref{NN}) and extended Nikov-Nikova method (Section \ref{ENN}), respectively, to generate four new schemes $\pi_{\rm lc},\pi_{\rm ps},\pi_{\rm is},\pi_{\rm cs}$ for $\Gamma_{\cdot Q}$. Let $z$ be the number of secrets that have been shared using $\pi_0$, $\ell$ be the storage occupied by every share, and $L(\pi)$ be the total storage occupied by all shares in an SSS $\pi$.
	
	\begin{table}[h]
		\centering
		\caption{The total storage ($L$) and the information rate ($\rho$)}
		\label{tbl1}
		\begin{tabular}{lcc}
			\toprule
			$\pi$ & $L(\pi)$ & $\rho(\pi)$\\
			\midrule
			$\pi_{\rm ps}$ & $n\ell z$ & 1 \\
			$\pi_{\rm is}$ & $(n+(n-m-1)m)\ell z$ & ${(m+1)^{-1}}$ \\
			$\pi_{\rm cs}$ & $(n+(n-t)m)\ell z$ & $\lceil {m(n-t+1)}/(n-m) +1\rceil^{-1}$\\
			$\pi_{\rm lc}$ & $(n-m)\ell z$ & 1\\
			\bottomrule
		\end{tabular}
	\end{table}

	In Table \ref{tbl1}, because $Q$ is unauthorized, we have that $m\le t-1$. It is not difficult to observe that when $m>0$, $L(\pi_{\rm lc})<L(\pi_{\rm ps})<L(\pi_{\rm cs})\le L(\pi_{\rm is})$ and $\rho(\pi_{\rm lc})=\rho(\pi_{\rm ps})>\rho(\pi_{\rm cs})\ge\rho(\pi_{\rm is})$. That is, our method gives the most storage-efficient and the most communication-efficient (i.e., highest information rate) scheme for $\Gamma_{\cdot Q}$ among the four methods.

 \vspace{1mm}
	\noindent {\bf Remark.}
	As a special case of HSS, a $d$-multiplicative secret sharing \cite{BIW10,TLM18} allows a user to share $d$ secrets among the universal set $\mathcal P$ of $n$ servers such that every server is able to locally convert its shares to a partial result and the sum of all server's partial results is equal to the product of the $d$ secrets. In particular, Barkol \cite{BIW10} has showed that Shamir's $(t,n)$-TSSS is $\lfloor n/t\rfloor$-multiplicative. Thus, the four schemes $\pi_{\rm lc},\pi_{\rm ps},\pi_{\rm is},\pi_{\rm cs}$ are $\lfloor (n-m-1)/(t-m)\rfloor$-multiplicative. Since $\lfloor (n-m-1)/(t-m)\rfloor\ge \lfloor n/t\rfloor$, the new schemes allow to homomorphically compute the product of more secrets.
 \vspace{1mm}
	
	\subsubsection{Performance Analysis}
	To confirm the advantages of our method in terms of storage and information rate, a simple experiment is performed.

 \begin{figure}[!t]
		\centering
		\subfloat[Storage]{\label{suba} \includegraphics[width=2.7 in]{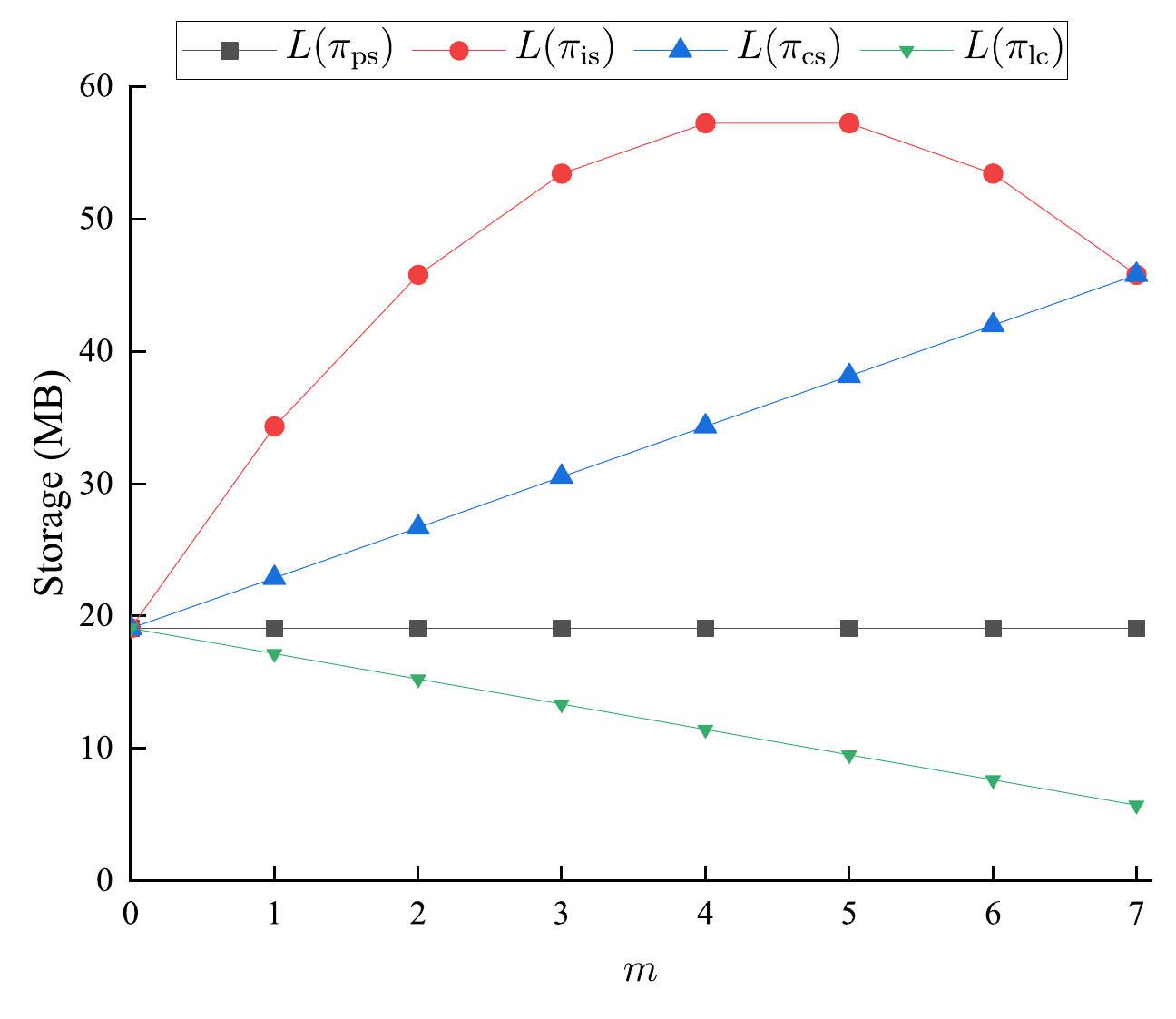}} \\
		\subfloat[Information rate]{\label{subb} \includegraphics[width=2.7 in]{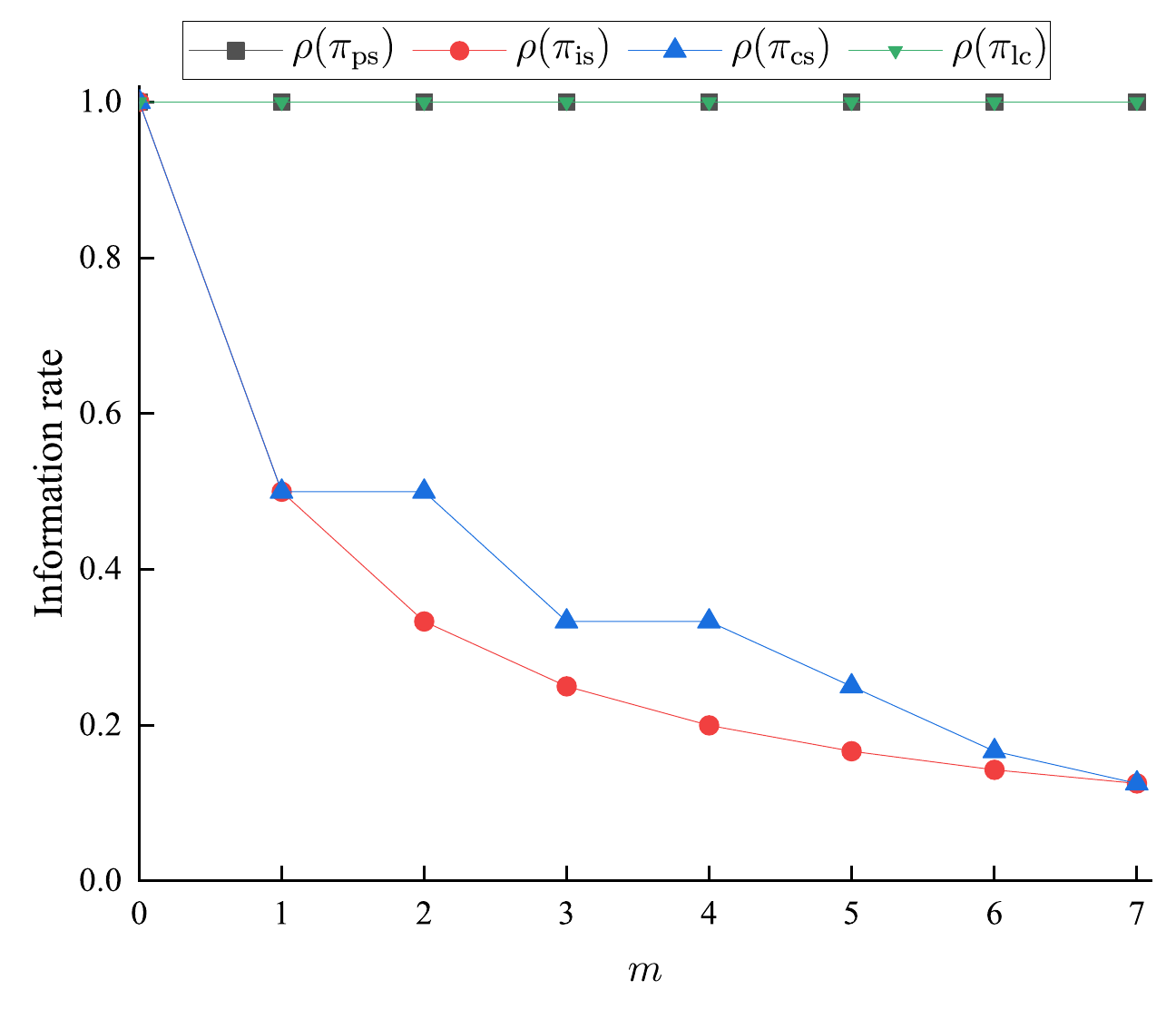}}
		\caption{The total {\bf storage} occupied by all shares in $\pi_{\rm ps},\pi_{\rm is},\pi_{\rm cs}$, our $\pi_{\rm lc}$ respectively and the {\bf information rate} of the four schemes, when $m$ servers are removed from $(8,10)$-threshold access structure. The storage occupied by every share is $16$ bit ($\ell=16$), and the number of secrets is $z=10^6$}
		\label{fig:t1}
	\end{figure}
	
Like \cite{LCH22,SAR18},  the users in our experiment are also allowed  to subscribe at most $10$ servers (i.e., $n=10$)  and we denote by $\mathcal P=\{P_1,\ldots,P_{10}\}$ the set of all servers. Let $p=2^{16}-15$ be a 16-bit prime. We set the threshold to be $t=8$ and construct Shamir's (8,10)-TSSS $M=(\mathbb{F}_p,\bm H,\bm t,\bm \psi)$ for the threshold access structure $\Gamma$ in the following way. We firstly choose $\bm H, \bm t$, and $\psi$ such that $\bm H=(\bm h_1,\ldots,\bm h_{10})^\top$, where $\bm h_i=(i^0,\ldots,i^{7})^\top$ for each $i\in [10]$, $\bm t=(1,0,\ldots,0)^\top\in \mathbb F_p^{8}$, and $\psi(i)=P_i$ for each $i\in[10]$. We randomly generate a list of {$10^6$} field elements as the secrets to be stored (i.e., the number of secrets is $z=10^6$). Let $Q\subseteq \mathcal P$ be a set of $m$ servers to be removed. The storage $\ell$ occupied by every share is $16$ bit.
	
	We have implemented the four methods given in Section \ref{sec:mcs} on a Dell OptiPlex 7050 Personal Computer that runs with an Intel Core i5-6500 (3.20GHz) processor and a RAM of 16 GB. We compare the four new schemes $\pi_{\rm ps},\pi_{\rm is},\pi_{\rm cs},\pi_{\rm lc}$ for $\Gamma_{\cdot Q}$ in terms of the total storage $L$ occupied by all shares (part (a) of Fig. \ref{fig:t1}) and the information rate (part (b) of Fig. \ref{fig:t1}) when the size $m$ of $Q$ increases from $0$ to $7$ with a step $1$.

	In part (a) of Fig. \ref{fig:t1}, when $n=10,t=8,\ell=16,z=10^6$, we have that
		$L(\pi_{\rm is})=-1.9073m^2+17.1661m+19.0735$, $L(\pi_{\rm cs})=3.8147m+19.0735$, $L(\pi_{\rm ps})=19.0735$, $L(\pi_{\rm lc})=-1.0973m+19.0735$.
	{Part (a) of Fig. \ref{fig:t1}} shows that our method occupies less storage than {others}, because our method can eliminate the storage occupied by {$Q$'s} shares. The more servers are removed, the more storage-effective our method is than {others}. 	In part (b) of Fig. \ref{fig:t1}, when $n=10,t=8$, we have that
	$
		\rho(\pi_{\rm lc})=\rho(\pi_{\rm ps})=1,\rho(\pi_{\rm cs})={\lceil{3m}/{(10-m)}+1\rceil}^{-1}, \rho(\pi_{\rm is})=({m+1})^{-1}
	$.
	Part (b) of Fig. \ref{fig:t1} shows that the information rate of $\pi_{\rm lc}$ and $\pi_{\rm ps}$ is higher than that of $\pi_{\rm is}$ and $\pi_{\rm cs}$.




	\section{Application in Single-cloud Storage}
	\label{sec:app}
	In an ABE scheme with an access structure $\Gamma$, the contraction of $\Gamma$ means reduction in the attribute requirements for decryption, so that more users will be allowed to access the encrypted data. In this section, we focus on contractions of access structures in CP-ABE schemes.
	
	\subsection{CP-ABE Model} \label{sec:CPABEmodel}
	A ciphertext-policy
attribute-based encryption (CP-ABE) scheme ${\sf (Setup, KeyGen, Encrypt, Decrypt)}$
consists
  of  four polynomial-time algorithms:
	    \begin{itemize}
	        \item ${\sf Setup}(\lambda,U)\rightarrow(\mathrm{PK},\mathrm{MSK}).$ The \emph{setup} algorithm takes a security parameter $\lambda$ and a universal set $U$ of attributes as input and outputs a public key $\mathrm{PK}$ and a master secret key $\mathrm{MSK}$.
	        \item ${\sf KeyGen}(\mathrm{PK},\mathrm{MSK},A)\rightarrow \mathrm{SK}.$ The
{\em key generation} algorithm takes the public key $\mathrm{PK}$, the master secret key $\mathrm{MSK}$, and a set $A$ of attributes as input and outputs a private key $\mathrm{SK}$ for $A$.
	        \item ${\sf Encrypt}(\mathrm{PK},M,\Gamma,\mathcal M)\rightarrow \mathrm{CT}.$ The \emph{encryption} algorithm takes the public key $\mathrm{PK}$, a message $M$, an access structure $\Gamma$ over $U$, and an LSSS $\mathcal M$ for $\Gamma$ as input, and outputs a ciphertext $\mathrm{CT}$
         such that  a user can  extract $M$ from $\mathrm{CT}$  if and only if its attributes form an authorized subset in $\Gamma$.
         It is assumed that $\mathrm{CT}$ implicitly includes $\Gamma$.
	        \item ${\sf Decrypt}(\mathrm{SK},\mathrm{CT})\rightarrow M.$ The \emph{decryption} algorithm takes as input a private key $\mathrm{SK}$ for a set $A$ of attributes and a ciphertext $\mathrm{CT}$, which includes an access structure $\Gamma$. If $A\in \Gamma$, it outputs a message $M$.
	    \end{itemize}

\begin{figure}[t]
		\centering
		\includegraphics[width=3.4in]{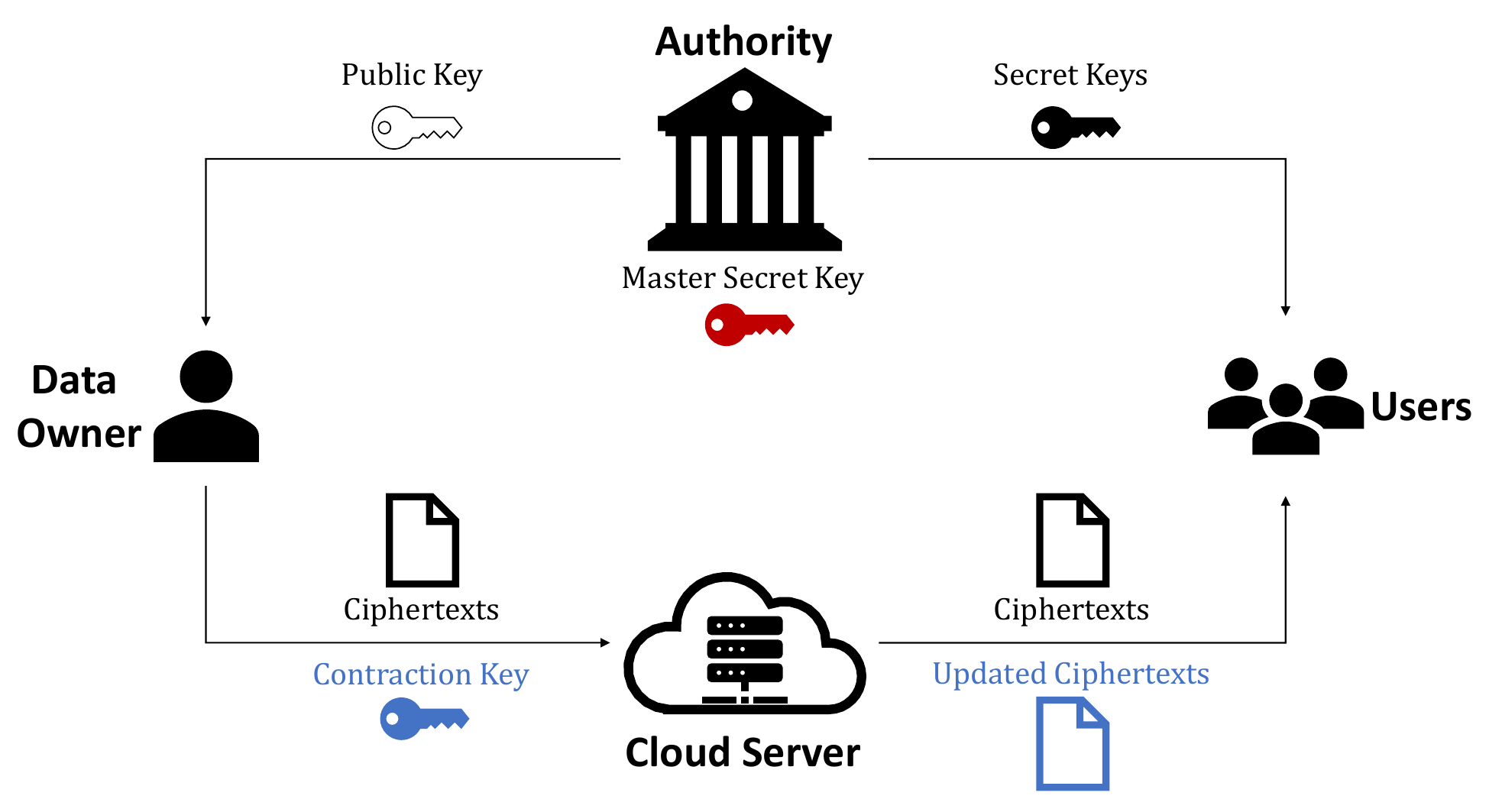}
		\caption{\bf CP-ABE (-CAS) Model}
		\label{fig:model}
	\end{figure}

    \noindent {\bf System architecture.} A CP-ABE scheme (depicted in Fig. \ref{fig:model}) involves four entities: {\em authority, data owner, server,} and {\em user}.
  The authority is  trusted and responsible to  run
     $\sf Setup$  to generate $(\rm PK,\mathrm{MSK})$  and run  $\mathsf{KeyGen}$   to
     generate a  private key $\rm SK$  for every   registered user.
By running  $\mathsf{Encrypt}$,  the {\em data owner}  may use $\mathrm{PK}$ to  encrypt its
data $M$ with an access policy $\Gamma$.
The server is  {\em honest-but-curious} and stores the resulting ciphertext ${\rm CT}$.
To learn $M$,  the user simply downloads $\rm CT$ from the server and runs $\mathsf{Decrypt}$.

\vspace{1mm}
	\noindent {\bf Security.} The security of CP-ABE schemes \cite{Wat11}  can be   defined with
   a security game $G_1$    between a challenger and an adversary and the game consists of the following phases:
	 \begin{itemize}
	     \item \textit{Setup:} The challenger runs the setup algorithm to generate
${\rm (PK,MSK)}$ and gives  $\mathrm{PK}$ to the adversary.
	     \item \textit{Phase 1:} The adversary queries the challenger for private keys corresponding to the attribute sets $S_1,\ldots,S_{q_1}$.
	     \item \textit{Challenge:} The adversary declares two equal length messages $M_0,M_1$ and an access structure $\Gamma^*$ such that none of the queried attribute sets $S_1,\ldots,S_{q_1}$ satisfies $\Gamma^*$. The challenger  chooses $b\in\{0,1\}$ randomly,  encrypts $M_b$ under $\Gamma^*$, and  gives the ciphertext $\mathrm{CT}^*$ the adversary.
	     \item \textit{Phase 2:} The adversary queries the challenger for private keys
corresponding to the  attribute sets $S_{q_1+1},\ldots,S_q$, with the restriction that none of these satisfies $\Gamma^*$.
	     \item \textit{Guess:} The adversary outputs a guess $b'$ for $b$.
	 \end{itemize}
	 The \emph{advantage} of the adversary in $G_1$ is defined as
 $\mathrm{Pr}[b'=b]-1/2$. A CP-ABE scheme is {\em secure} if all PPT
 adversaries have at most a negligible advantage in $G_1$.

	\subsection{Waters' CP-ABE Scheme} \label{sec:watscheme}
 Waters  \cite{Wat11} constructed a CP-ABE scheme that is secure under the
 decisional parallel bilinear diffie-hellman exponent assumption for bilinear groups
 (see Appendix \ref{app:bp} for definition of bilinear groups). Their scheme can be detailed   as follows: 	
	\begin{itemize}
	        \item ${\sf Setup}(\lambda,U)$. Choose  a bilinear group $\mathbb G=\langle  g\rangle$ of prime order $p$  and a bilinear map $e:\mathbb G\times \mathbb G\rightarrow \mathbb G_T$.
Choose random exponents $\beta,a\in\mathbb{Z}_p$. For every attribute $x\in U$, choose a random value $
	        T_x\in \mathbb G$. Output
$$\mathrm{MSK}=g^\beta,\mathrm{PK}=\{g,g^a,e(g,g)^\beta,(T_x)_{x\in U}\}.$$
	        \item ${\sf KeyGen}(\mathrm{PK},\mathrm{MSK},A)$.
 Choose  a random $t\in \mathbb Z_p$. Compute
	        $K=g^\beta g^{at},L=g^t,K_x=T_x^t
	        $
	        for all $x\in A$. Output
$$\mathrm{SK}=\{A,K,L,(K_x)_{x\in A}\}.$$
	
	        \item ${\sf Encrypt}(\mathrm{PK},M,\Gamma,\mathcal M)$. Parse $\mathcal M$ as $(\mathbb Z_p,\bm H,\bm t,\psi)$, an LSSS for access structure $\Gamma$, where $\bm H=(\bm h_1,\ldots,\bm h_\ell)^\top$ is an $\ell\times d$ matrix over $\mathbb Z_p$, $\bm t=(1,0,\ldots,0)^\top\in \mathbb Z_p^d$ is a target vector, $\psi:[\ell]\rightarrow U$ is a map from each row $\bm h_i^\top$   to an attribute $\psi(i)$. Choose  a random vector $\bm v=(s,v_2,\ldots,v_d)\in \mathbb Z_p^d$. For each $i\in[\ell]$, compute $s_i=\bm h_i^\top \bm v$ and  chooses a random value $r_i\in \mathbb Z_p$.
Let 	            $C=Me(g,g)^{\beta s},C'=g^s,C_i=g^{as_i}T_{\psi(i)}^{-r_i},D_i=g^{r_i}$
	        for all $i\in [\ell]$. Finally, output
$$\mathrm{CT}=\{\mathcal M,C,C',(C_i,D_i)_{i\in [\ell]}\}.$$
	
	        \item ${\sf Decrypt}(\mathrm{SK},\mathrm{CT})$. Suppose $A\in \Gamma$. Compute
   the constants $\{\alpha_i\in \mathbb Z_p:\psi(j)\in A\}$ such that $\sum_{\psi(j)\in A}\alpha_j\bm h_j^\top=\bm t$. Compute
	        $$\frac{e(C',K)}{\prod_{\psi(j)\in A}\left(e(C_j,L)e(D_j,K_{\psi(j)})\right)^{\alpha_j}}=e(g,g)^{\beta s}.$$
	        Finally, output $M=C/e(g,g)^{\beta s}$.

	    \end{itemize}
For every  $i\in [\ell]$, the component $C_i$ of the ciphertext $\mathrm{CT}$ is associated with
 an  attribute $\psi(i)$ and provides necessary \emph{decrypting information} to an authorized
attribute set $A$ that contains  $\psi(i)$.
As a result,  the problem of  eliminating the control of $\psi(i)$ (in general, an unauthorized
subset of attributes)   over decryption is reduced to the problem (p1).

\subsection{CP-ABE with Contractions of Access Structure}

To enable   contractions of  access structures, we extend the   CP-ABE model
of Section \ref{sec:CPABEmodel} to
a new model of  {\em CP-ABE with contractions of access structures  (CP-ABE-CAS)}.
The new model   ${\sf (Setup, KeyGen, Encrypt^*, Decrypt, Contract)}$
is obtained from that of CP-ABE by making two changes: (1) enhancing the
$\sf Encrypt$ in CP-ABE to a new encryption algorithm ${\sf Encrypt}^*$ that also outputs a contraction
key $\rm CK$; (2)  adding a new contraction algorithm $\sf Contract$ that allows one to use
 $\rm CK$ (or its restrictions) to contract the access structure associated with a ciphertext.
More precisely,  the new algorithms ${\sf Encrypt}^*$ and
$\sf Contract$  can be detailed as below:
\begin{itemize}
            \item ${\sf Encrypt}^*(\mathrm{PK},M,\Gamma,\mathcal M)\rightarrow
            (\mathrm{CT},\mathrm{CK}).$ The {\em modified encryption}  algorithm
             additionally outputs a contraction key $\mathrm{CK}$, which
               implicitly includes $U$. For any $Q\subseteq U$, $\mathrm{CK}$ can be restricted to $\mathrm{CK}_Q$.
	        \item  ${\sf Contract}(\mathrm{PK},\mathrm{CT},Q,\mathrm{CK}_Q)
\rightarrow \mathrm{CT}_{\cdot Q}.$ The {\em contraction} algorithm takes
the public key $\mathrm{PK}$, a ciphertext $\mathrm{CT}$ with an
access structure $\Gamma$, a set $Q\in 2^U\setminus \Gamma$, and a contraction
 key $\mathrm{CK}_Q$ restricted to $Q$ as input and outputs
 a new ciphertext $\mathrm{CT}_{\cdot Q}$ that can be decrypted by any authorized
  attribute set in $\Gamma_{\cdot Q}$. Similarly, $\mathrm{CT}_{\cdot Q}$ includes $\Gamma_{\cdot Q}$.
	    \end{itemize}

      \noindent {\bf System architecture.}
Referring to Fig. \ref{fig:model}, to enable contractions  of access structures with CP-ABE-CAS,
a data owner may run
${\sf Encrypt}^*(\mathrm{PK},M,\Gamma,\mathcal M)$ to
produce ${(\rm CT, CK)}$, store $\rm CT$ on the
 cloud server and keep  $\rm CK$ secret in local storage such that it is the only one that can later
   request the server to contract
       access structures.
To  remove an attribute $y$ ($Q=\{y\}\in 2^U\setminus \Gamma$), the data owner may send
a   restricted contraction key  $\mathrm{CK}_Q$  to the
       server and let the server   execute
       ${\sf Contract}(\mathrm{PK},\mathrm{CT},Q,\mathrm{CK}_Q)$.
Afterwards, all users with an authorized attribute set in
         $\Gamma_{\cdot Q}$ will be able to decrypt the contracted ciphertext
$\mathrm{CT}_{\cdot Q}$.

\vspace{1mm}
  \noindent {\bf Security.}
We define the security of CP-ABE-CAS  with
   a security game $G_2$ that
     consists of the following phases:
	    \begin{itemize}
  \item \textit{Setup:} The challenger runs the setup algorithm to generate
${\rm (PK,MSK)}$ and gives  $\mathrm{PK}$ to the adversary.
		     \item \textit{Phase 1:} The adversary queries the challenger for private keys corresponding to the attribute sets $S_1,\ldots,S_{q_1}$.
	        \item \textit{Challenge:} The adversary declares two equal length
messages $M_0,M_1$, an access structure $\Gamma^*$, an LSSS $\mathcal M$ for $\Gamma^*$,
 and an unauthorized subset $Q$ of $\Gamma^*$, where $\Gamma^*_{\cdot Q}$ cannot be satisfied by any
  of  $S_1,\ldots,S_{q_1}$. The challenger chooses $b\in\{0,1\}$ randomly,
  encrypts $M_b$ under $\Gamma^*$, producing $\mathrm{CT}^*$,
  runs  ${\sf Contract}(\mathrm{PK},\mathrm{CT}^*,Q,\mathrm{CK}_Q)$ to
   output $\mathrm{CT}^*_{\cdot Q}$, and  gives $(\mathrm{CT}^*,\mathrm{CT}^*_{\cdot Q},
   \mathrm{CK}_Q)$ to the adversary.
	     \item \textit{Phase 2:} The adversary queries the challenger for private
 keys corresponding to the  attribute sets $S_{q_1+1},\ldots,S_q$, with the
  restriction that none of these satisfies $\Gamma^*_{\cdot Q}$.
	     \item \textit{Guess:} The adversary outputs a guess $b'$ for $b$.
	    \end{itemize}
	    The {\em advantage} of an adversary in $G_2$  is $\mathrm{Pr}[b'=b]-1/2$.
A CP-ABE-CAS is {\em secure} if all PPT adversaries have at most a negligible
advantage in $G_2$.

	    \subsection{Our CP-ABE-CAS scheme} \label{sec:cpabecas_scheme}

     In this section, we upgrade the CP-ABE scheme of Waters \cite{Wat11}
     (see Section \ref{sec:watscheme})
     to a CP-ABE-CAS scheme with specific constructions of
        ${\sf Encrypt}^*$ and   $\sf Contract$.
        Our contraction algorithm will be constructed based on the  transformations from
        Section \ref{sec:our transformations}. In particular, we will consider
        contractions of ideal access structures at an unauthorized attribute set
         $Q$ of cardinality 1, because the algorithm can be easily extended to the case   $|Q|>1$.
The details of  ${\sf Encrypt}^*$ and ${\sf Contract}$ are described as follows:
   \begin{itemize}
            \item ${\sf Encrypt}^*(\mathrm{PK},M,\Gamma,\mathcal M).$
Execute the encryption algorithm  ${\sf Encrypt}(\mathrm{PK},M,\Gamma,\mathcal M)$.
Output the contraction key
 $\mathrm{CK}=\{r_i: i\in[\ell]\}.$  For any $Q\subseteq U$,
let $\mathrm{CK}_Q=\{r_i: i\in \psi^{-1}(Q)\}$ be the restricted contraction key.
	        \item ${\sf Contract}(\mathrm{PK},\mathrm{CT},Q,\mathrm{CK}_Q).$ Suppose
that $Q=\{y\}\in 2^U\setminus \Gamma$. W.l.o.g, assume that $\psi^{-1}(y)=\ell$.
Invoke  Algorithm \ref{alg:s1} with $(\mathcal M,Q)$ as input to
generate an MSP $\mathcal M'$ for $\Gamma_{\cdot Q}$.
Let $k$ be the integer chosen at step 1 of Algorithm \ref{alg:s1}.
For all $i\in[\ell-1]$, compute        $$C'_i=
	        C_i(C_{\psi^{-1}(y)}T_y^{r_{\psi^{-1}(y)}})^{-\frac{h_{ik}}{h_{\ell k}}}.$$
Finally, output
	        $\mathrm{CT}_{\cdot Q}=\{\mathcal M',C,C',(C'_i,D_i)_{i\in [\ell-1]}\}$.
	    \end{itemize}

Our contraction algorithm   properly integrates the decrypting information associated with $Q$ into
that   associated with every remaining attribute.
As the new ciphertext ${\rm CT}_{\cdot Q}$ is \emph{shorter} than the original ciphertext
$\rm CT$,  hereafter we denote  by  $\mathsf{Contract}_{\rm sct}$ this contraction algorithm.

\vspace{1mm}
\noindent
{\bf Correctness and security.}
The correctness of our  scheme  follows from that
of  Waters' scheme (Section \ref{sec:watscheme}) and  Theorem \ref{1}.
The security of our  scheme is an easy extension to that of  Waters' scheme and appears in
 Appendix \ref{appendix:CPABEsecur}.

\vspace{1mm}
			\noindent {\bf Client-side storage.}
In our CP-ABE-CAS,
 the  client needs to   store a contraction key $\mathrm{CK}$ whose length is $\ell$ times that
 of a single message, where $\ell$ is the total number of attributes.
 When $|\rm CK|$ occupies more storage than the outsourced data,
   storing data with a cloud server will become meaningless.
   This concern can be easily relieved by the data owner choosing a
   PRF $F_\gamma:\{0,1\}^*\rightarrow \mathbb{Z}_p$
   and generating every element $r_i$ in $\rm CK$
   as $r_i=F_\gamma(i)$. The data owner only needs to keep
   the secret key $\gamma$ of the PRF as a
   {\em long-term} contraction key.


 \vspace{1mm}
   \noindent {\bf Remark.}
		If we apply Martin's method or Nikov-Nikova method from Section \ref{sec:mcs} to
construct  CP-ABE-CAS, then intuitively $\mathsf{Contract}$ may realized by  {\em appending
  the restricted contraction key to the original ciphertext}. That is, $\mathsf{Contract}$
   simply outputs $\mathrm{CT}_{\cdot Q}=\{\mathcal M,C,C',(C_i,D_i)_{i\in [\ell]},\mathrm{CK}_Q\}$.
   Then an authorized attribute set  in $\Gamma_{\cdot Q}$ can leverage $\mathrm{CK}_Q$ to compute
		$e(C_j,L)e(D_j,K_{\psi(j)})=e(g,g)^{a s_j t}=e(C_jT_{\psi(j)}^{r_j},L)$
	for all $j\in\psi^{-1}(Q)$ and then recover the message $M$.
This straightforward contraction algorithm via \emph{ciphertext extension}
is referred to as  $\mathsf{Contract}_{\rm ect}$.

	\subsection{Performance Analysis}
	In this section, in order to compare the trivial solution (referred to as $\mathsf{Contract}_{\rm re}$) mentioned in Section \ref{ss:apps}, $\mathsf{Contract}_{\rm ect}$ and $\mathsf{Contract}_{\rm sct}$ in Section \ref{sec:cpabecas_scheme}, we have implemented all algorithms on a Dell PowerEdge T640 Server that runs with an Intel Xeon Gold 5218 (2.30GHz) processor and a RAM of 16GB. We have used the Paring-Based Cryptography (PBC) library and the fast library for number theory (FLINT) to implement the algorithm. We use the type A elliptic curve $y^2 = x^3 + x$ and the order of the bilinear group is a 160-bit prime $p=2^{159}+162259276829213363391578010288129$.
	
	We compare three solutions in terms of the following measures: during the contraction process, the communication and computation cost for the contraction operation; after contraction, the storage cost for the contracted ciphertext, the communication cost for downloading ciphertext, the computation cost for the user decryption. All the following experiments are for a  {\em single} message. We refer to the parameter settings of \cite{GSB+21}.
	For the contraction process, the main factor of the difference between the three solutions is the number $n$ of attributes in the universe set $U$. We use the threshold access structure to encrypt the message, set the threshold to be $t=8$ and set $n=|U|$ to be from $10$ to $100$ with a step $10$.
	For the performance of the contracted ciphertext, the main factor of the difference between the three solutions is the number $m$ of the attributes to be removed. We use a $(t,n)$-threshold structure where $t=8,n=10$ and choose the unauthorized subset $Q$ of attributes to be removed such that $m=|Q|$ is from $0$ to $7$ with a step $1$.
	Furthermore, to compute the execution time of the decryption for the original ciphertext and the contracted ciphertext, we simply allow the client to possess a set $A$ of attributes such that $|A|=t$ and introduce a user whose attribute set is $B$ such that $|B|=t-m$. To make the execution time more accurate, we execute each experiment 1000 times to compute an average time.

	\subsubsection{Evaluation of the Contraction Process}

 \begin{figure}[!t]
	\centering
    \includegraphics[width=2.7 in]{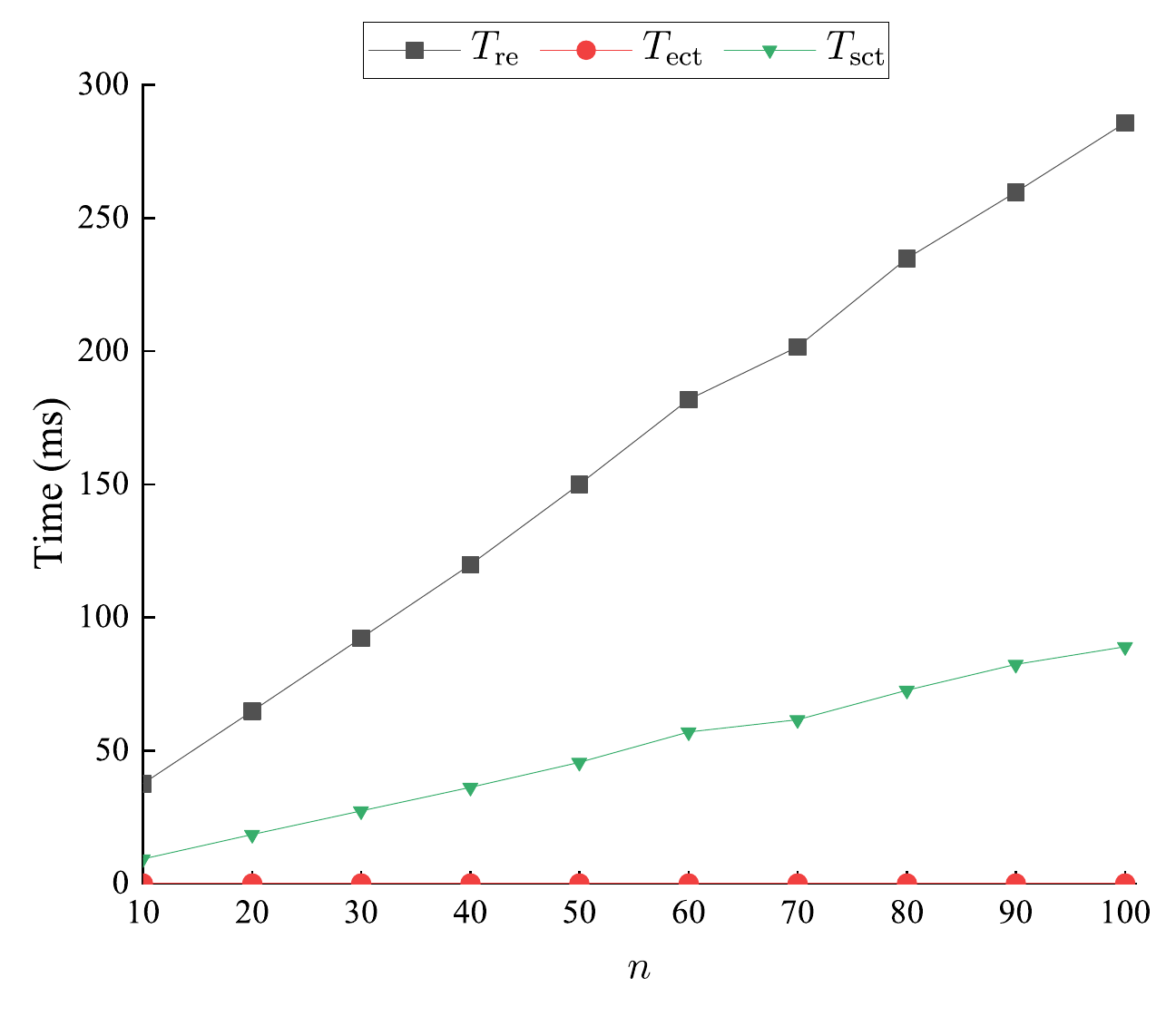}
    \caption{{\bf Computation of contraction.} The execution time (resp. $T_{\rm re}$, $T_{\rm ect}$, and $T_{\rm sct}$) of $\mathsf{Contract}_{\rm re}$, $\mathsf{Contract}_{\rm ect}$ and our $\mathsf{Contract}_{\rm sct}$ when $1$ attribute is removed from $(8,n)$-threshold access structure (where $n=10,20,\ldots,100$)}
    \label{fig:abeeff}
    \end{figure}

\begin{figure}[!t]
		\centering
		\subfloat[Server-side communication]{\includegraphics[width=2.7 in]{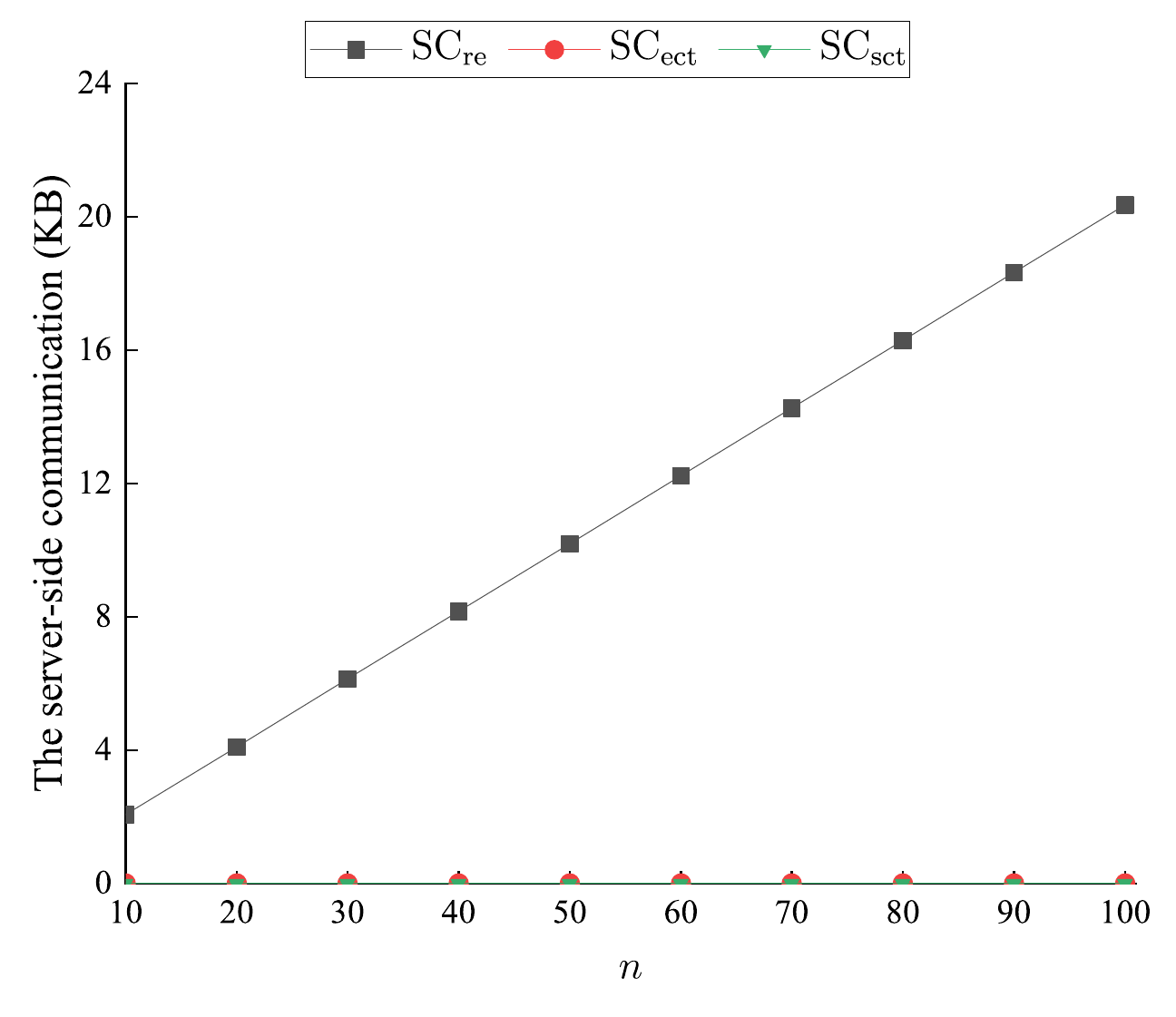}} \\
		\subfloat[Client-side communication]{\includegraphics[width=2.7 in]{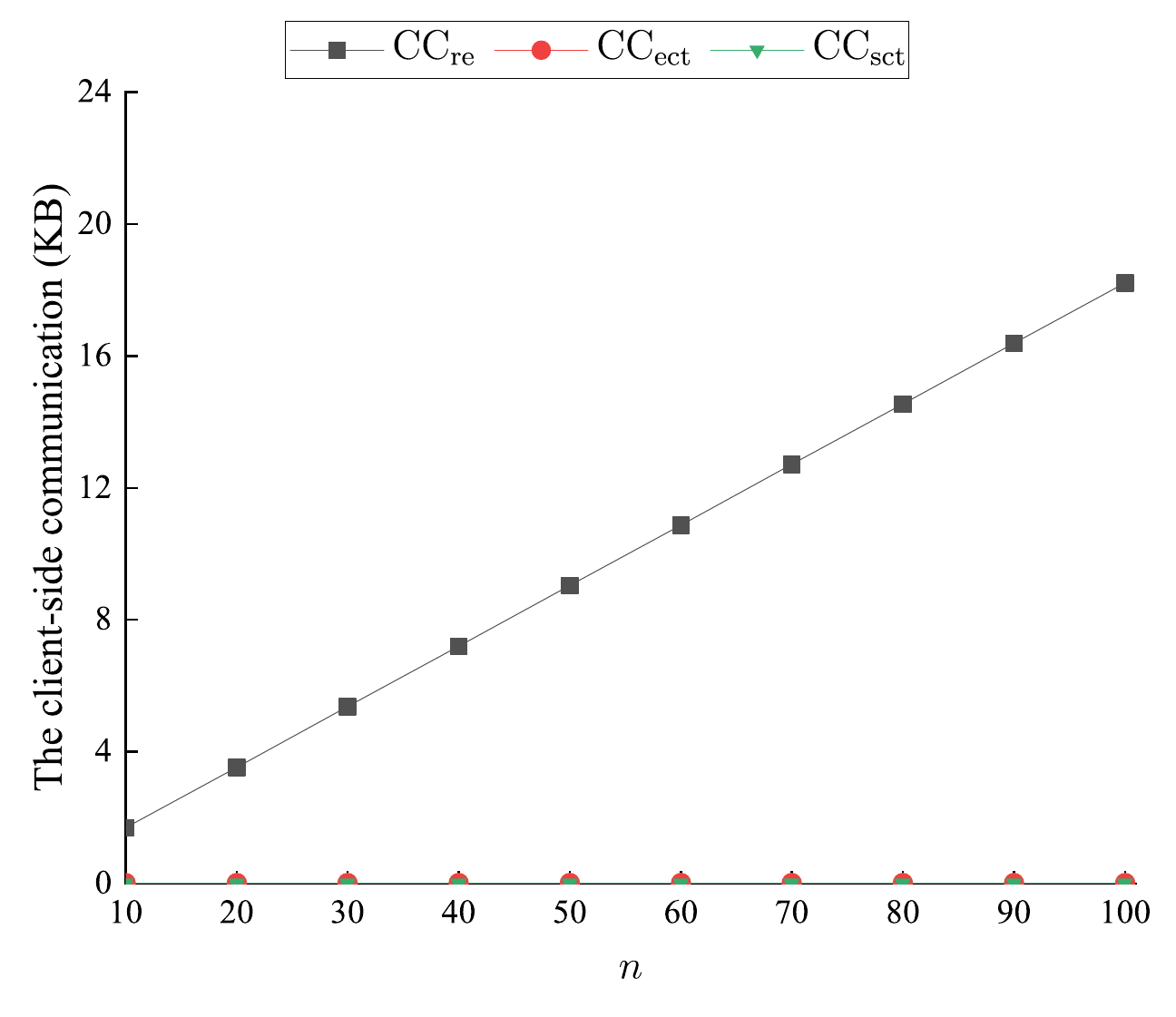}}
		\caption{{\bf Communication of contraction.} The server-side (resp. ${\rm SC}_{\rm re}$, ${\rm SC}_{\rm ect}$, and ${\rm SC}_{\rm sct}$) and client-side communication (resp. ${\rm CC}_{\rm re}$, ${\rm CC}_{\rm ect}$, and ${\rm CC}_{\rm sct}$) of $\mathsf{Contract}_{\rm re}$, $\mathsf{Contract}_{\rm ect}$ and our $\mathsf{Contract}_{\rm sct}$ when $1$ attribute is removed from $(8,n)$-threshold access structure (where $n=10,20,\ldots,100$)}
		\label{fig:communication}
	\end{figure}
\

	\noindent \textbf{Communication cost.}
	In the trivial solution $\mathsf{Contract}_{\rm re}$, to realize the contraction of the access structure, provided that the client also has access to its data on the server and its attribute set is $A$, it should download and decrypt the ciphertext from the server and then upload new ciphertext after and re-encrypting.
	While in $\mathsf{Contract}_{\rm sct}$ and $\mathsf{Contract}_{\rm ect}$, it only requires the client to send the contraction key $\mathrm{CK}_Q$. Theoretically, regardless of client-side or server-side communication, $\mathsf{Contract}_{\rm sct}$ and $\mathsf{Contract}_{\rm ect}$ are superior to $\mathsf{Contract}_{\rm re}$. This is confirmed by the following results.
	 Fig.\ \ref{fig:communication} shows the server-side and client-side communication of $\mathsf{Contract}_{\rm re}$ (${\rm SC}_{\rm re}$ and ${\rm CC}_{\rm re}$), $\mathsf{Contract}_{\rm ect}$ (${\rm SC}_{\rm ect}$ and ${\rm CC}_{\rm ect}$) and $\mathsf{Contract}_{\rm sct}$ (${\rm SC}_{\rm sct}$ and ${\rm CC}_{\rm sct}$) when the threshold of access structure is $t=8$ and the number of the attributes to be removed is $m=1$. In particular,
		${\rm CC}_{\rm re}=0.1836n-0.1445,
		{\rm CC}_{\rm ect}={\rm CC}_{\rm sct}=0.0195, {\rm SC}_{\rm re}=0.2031n+0.0391, {\rm SC}_{\rm ect}={\rm SC}_{\rm sct}=0$.
	From the results, it follows that whether on the client side or on the server side, the communication cost of $\mathsf{Contract}_{\rm sct}$ and $\mathsf{Contract}_{\rm ect}$ is much lower than $\mathsf{Contract}_{\rm re}$.



	\vspace{1mm}
	\noindent \textbf{Computation cost.}
	Here we compare the execution time (resp. $T_{\rm re}$, $T_{\rm ect}$, and $T_{\rm sct}$) of three algorithms $\mathsf{Contract}_{\rm re}$, $\mathsf{Contract}_{\rm ect}$, and $\mathsf{Contract}_{\rm sct}$. Note that we ignore the time consumed by the communication between the server and the client, thus $\mathsf{Contract}_{\rm ect}$ takes negligible time in the contraction process.
	As shown in Fig. \ref{fig:abeeff}, when the threshold and the number of attributes to be removed are fixed, the execution time of $\mathsf{Contract}_{\rm re}$ and $\mathsf{Contract}_{\rm sct}$ is almost linear with the number $n$ of attributes. In particular, when $t=5,m=1$, we can obtain by linear fitting that
		$T_{\rm re}=2.7802n+9.8260,T_{\rm sct}=0.8964n+0.5536,T_{\rm ect}=0$.
	Theoretically, as the number of attributes grows, more pairs $(C_i,D_i)_{i\in[|U|-1]}$ in the encryption algorithm and more values $\{C_i':i\in[|U|-1]\}$ in the contraction algorithm need to be calculated. The results confirm it. Precisely, the execution time of $\mathsf{Contract}_{\rm re}$ is roughly $3$ times that of $\mathsf{Contract}_{\rm sct}$. This shows that $\mathsf{Contract}_{\rm sct}$ is more efficient than $\mathsf{Contract}_{\rm re}$.

	\subsubsection{Evaluation of the Contracted Ciphertext}
\

	\noindent \textbf{Storage cost.}
	We compare the length (resp. $l_{\rm re}$, $l_{\rm ect}$, and $l_{\rm sct}$)  of the contracted ciphertext obtained by three algorithms $\mathsf{Contract}_{\rm re}$, $\mathsf{Contract}_{\rm ect}$, and $\mathsf{Contract}_{\rm sct}$ in Fig. \ref{fig:abect}. The length $l_{\rm ect}$ of the contracted ciphertext generated by $\mathsf{Contract}_{\rm ect}$ is positively correlated with the size $m$ of $Q$, while the length $l_{\rm sct}$ (and $l_{\rm re}$) of the ciphertext generated by $\mathsf{Contract}_{\rm sct}$ (and $l_{\rm re}$) is negatively correlated with $m$. In particular, when $n=10,t=8$, we have that
		$l_{\rm ect}=0.0195 m+2.0703,
		l_{\rm re}=l_{\rm sct}=0.0195 m^2-0.3984 m+2.0703$.
	Theoretically, as the size of $Q$ grows, the length of the contraction key $\mathrm{CK}_{\cdot Q}$ included in the contracted ciphertext generated by $\mathsf{Contract}_{\rm ect}$ increases, while in the contracted ciphertext generated by $\mathsf{Contract}_{\rm sct}$ (and $\mathsf{Contract}_{\rm re}$), the number of the pairs $(C_i',D_i)_{i\in [|U|-|Q|]}$, the size of the matrix $\bm H$ and the map $\psi$ in the LSSS $\mathcal M$ all decrease. Thus, $\mathsf{Contract}_{\rm sct}$ (and $\mathsf{Contract}_{\rm re}$) is more cost-efficient than $\mathsf{Contract}_{\rm ect}$ in terms of storage.
	
	\begin{figure}[t]
		\centering
		\includegraphics[width=2.7in]{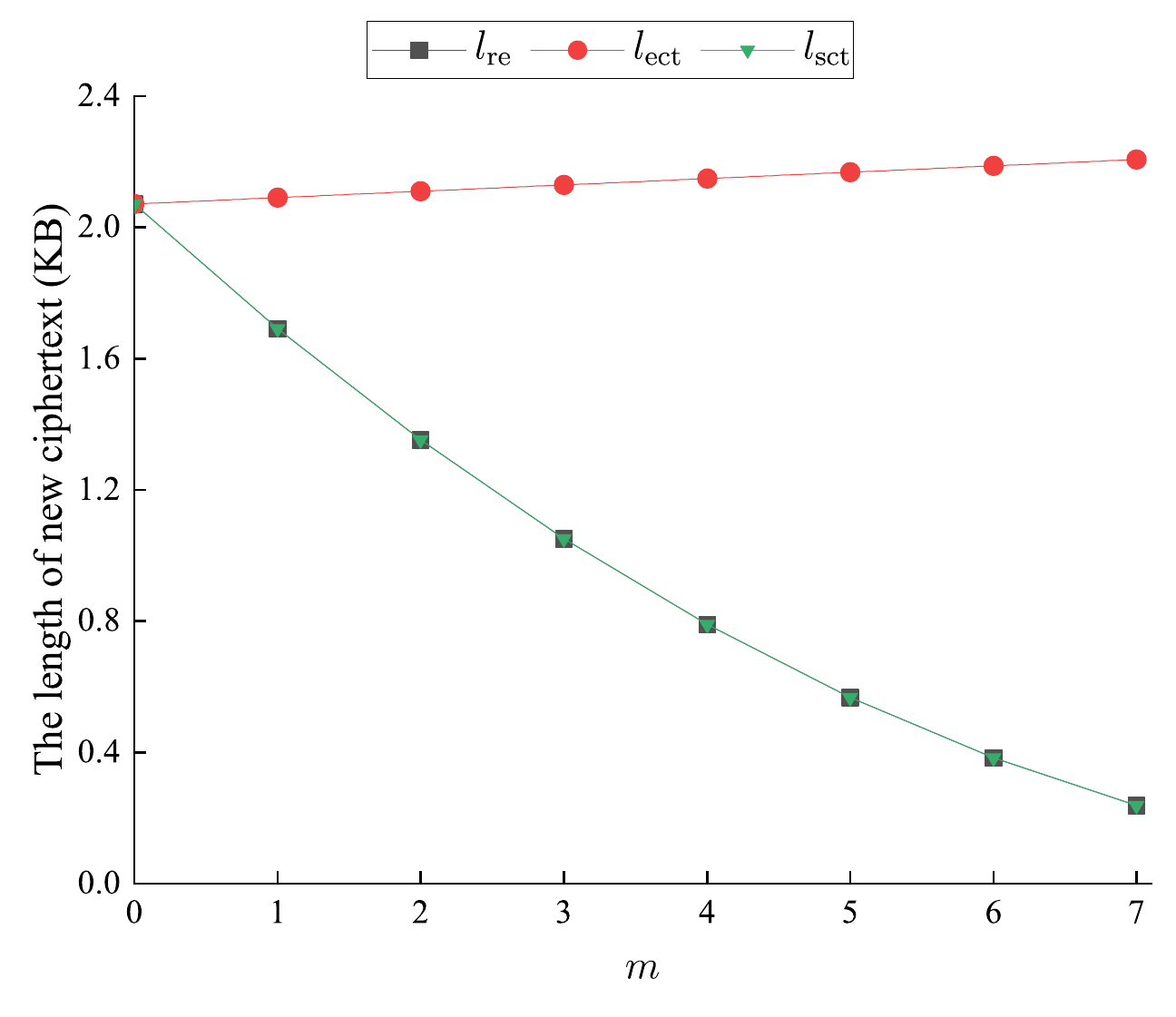}
		\caption{{\bf Storage of contracted ciphertext (implies communication of downloading contracted ciphertext).} The length (resp. $l_{\rm re}$, $l_{\rm ect}$, and $l_{\rm sct}$) of the contracted ciphertext obtained by $\mathsf{Contract}_{\rm re}$, $\mathsf{Contract}_{\rm ect}$, and our $\mathsf{Contract}_{\rm sct}$ after $m$ attributes are removed from $(8,10)$-threshold access structure (where $m=1,2,\ldots,7$)}
		\label{fig:abect}
	\end{figure}

	\vspace{1mm}
	\noindent \textbf{Communication cost.}
	During the user decryption process, the server sends the ciphertext to the user, and the server-side communication is exactly the length of the ciphertext, so the result is same as storage cost. Thus, $\mathsf{Contract}_{\rm sct}$ and $\mathsf{Contract}_{\rm re}$ are more efficient in terms of communication for the decryption of the contracted ciphertext than $\mathsf{Contract}_{\rm ect}$.

 \begin{figure}[!t]
		\centering
		\includegraphics[width=2.7in]{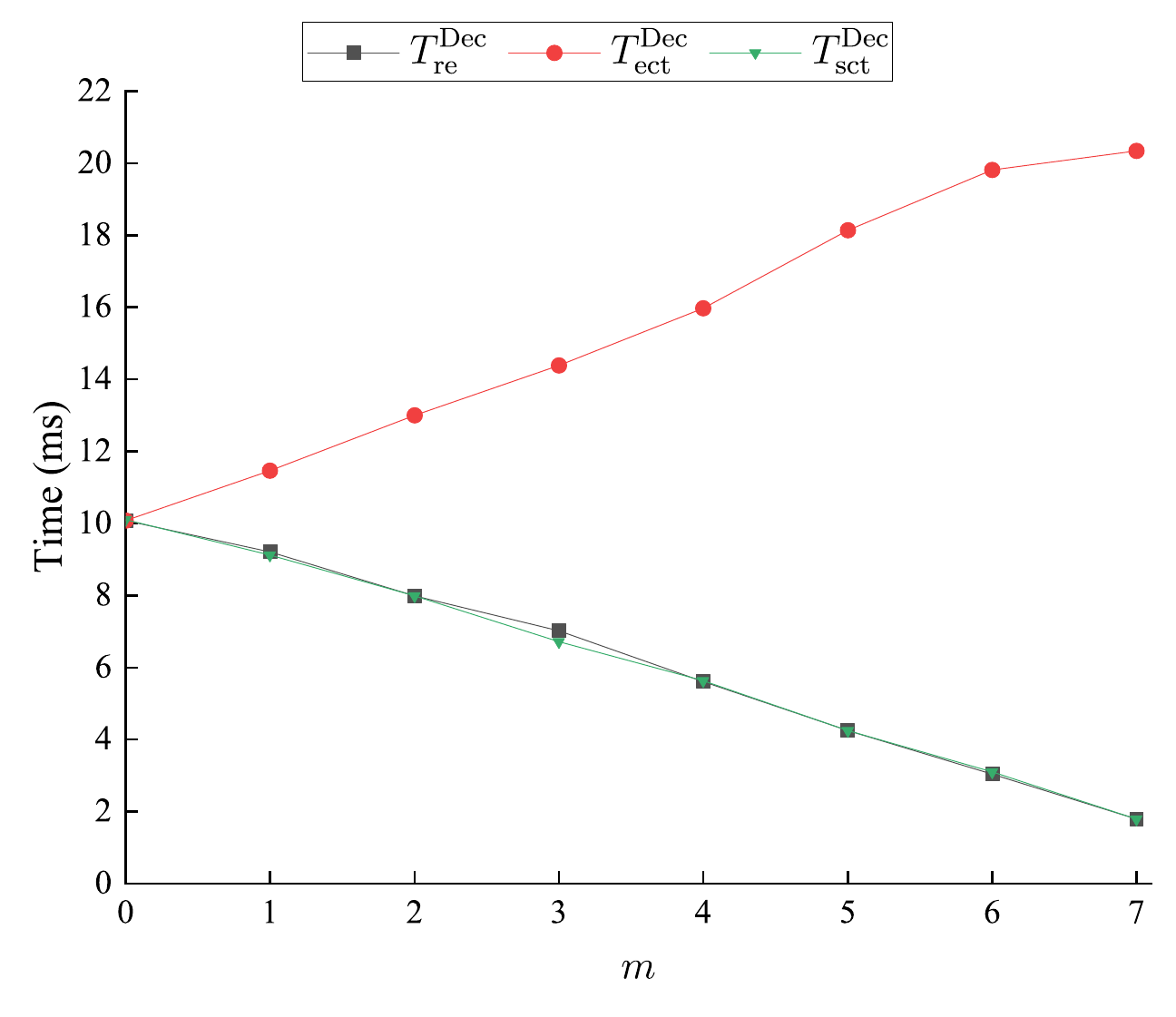}
		\caption{{\bf Computation of decrypting contracted ciphertext.} The execution time (resp. $T_{\rm re}^{\rm Dec}$, $T_{\rm ect}^{\rm Dec}$, and $T_{\rm sct}^{\rm Dec}$) of decryption for the contracted ciphertext obtained by $\mathsf{Contract}_{\rm re}$, $\mathsf{Contract}_{\rm ect}$ and $\mathsf{Contract}_{\rm sct}$ after $m$ attributes are removed from $(8,10)$-threshold access structure (where $m=1,2,\ldots,7$)}
	\label{fig:dectime}
	\end{figure}

	\vspace{1mm}
	\noindent \textbf{Computation cost.}
	As shown in Fig. \ref{fig:dectime}, when the universe set $U$ of attributes is set such that $n=|U|=10$ and the threshold is set to be $t=8$, the execution time $T_{\rm ect}^{\rm Dec}$ of the decryption for the contracted ciphertext obtained by $\mathsf{Contract}_{\rm ect}$ is positively correlated with the number $m$ of attributes to be removed, while the execution time $T_{\rm sct}^{\rm Dec}$ and $T_{\rm re}^{\rm Dec}$ of the decryption for the contracted ciphertext obtained by $\mathsf{Contract}_{\rm sct}$ and $\mathsf{Contract}_{\rm re}$ are negatively correlated with the number $m$ of attributes to be removed.
	In particular, when $n=10,t=8$, we can obtain by linear fitting that
		$T_{\rm ect}^{\rm Dec}=1.5540m+9.9558,
		T_{\rm re}^{\rm Dec}=-1.2073m+10.3447,
		T_{\rm sct}^{\rm Dec}=-1.1978m+10.2759$.
	Thus, the computation cost of the user-side decryption for the contracted ciphertext obtained by $\mathsf{Contract}_{\rm sct}$ and $\mathsf{Contract}_{\rm re}$ is lower than the computation cost of the user-side decryption for the contracted ciphertext obtained by $\mathsf{Contract}_{\rm ect}$.

	\subsubsection{Advantages of Our Contraction Algorithm}
	Consider that in many practical scenarios the number of users may be quite large, our contraction algorithm $\mathsf{Contract}_{\rm sct}$ allows the server to update the ciphertext only once when removing attributes, then the contracted ciphertext may be downloaded and decrypted by a large number of users. For example, when the number of attributes in $U$ is $n=10$, the threshold is $t=8$, the number of attributes to be removed is $m=7$, the execution time of $\mathsf{Contract}_{\rm sct}$ is roughly $10$ ms. For every user, the execution time $T_{\rm sct}^{\rm Dec}$ of the decryption for the contracted ciphertext obtained by $\mathsf{Contract}_{\rm sct}$ is roughly $2$ ms, while the execution time $T_{\rm ect}^{\rm Dec}$ of the decryption for the contracted ciphertext obtained by $\mathsf{Contract}_{\rm ect}$ is roughly $20$ ms. Note that the time of updating the ciphertext is only about $10$ ms, so when the number of users is quite large, it is clear that our algorithm is superior to the algorithm with ciphertext extension in terms of chronic computation cost.
	Furthermore, the length of the contracted ciphertext obtained by $\mathsf{Contract}_{\rm sct}$ is less than that obtained by $\mathsf{Contract}_{\rm ect}$. This not only saves the server-side storage, but also saves the server-side communication cost for users to download the ciphertext.
	In conclusion, our algorithm takes a certain amount of time to update the ciphertext, but optimizes the server-side storage and the overall communication and the user-side computation when a large number of users request new ciphertext.

	\section{Concluding Remarks}
	\label{sec:concluding remarks}
	In this paper, we  proposed algorithms that can efficiently
 transform  a given LSSS for an access structure  to
 LSSSs  for contractions of  the access structure.
 We also show their applications in solving the data relocating problem in
 multi-cloud storage and the   attribute removal problem
  in the CP-ABE based single-cloud storage.
Our solutions are storage efficient and assume honest-but-curious cloud servers.
It remains open to consider malicious servers and also
 ensure the integrity of the cloud data against malicious servers.
  In Appendix \ref{DI}, we briefly discuss several existing
         techniques that may be used to solve the data integrity problem.

 \section*{Acknowledgments}
 This work was supported in part
by the National Natural Science Foundation of China (No. 62372299) and the Natural Science Foundation of Shanghai (No. 21ZR1443000).



\clearpage

\appendices

\section{Proof for Lemma 1}
\label{app:lem1}

 As ${\rm rank}({\bm H}_Q)=r$, there exist $r$ rows of ${\bm H}_Q$ that are
 linear independent over $\mathbb{F}$. Suppose that these rows are
 $({\bm h}_{w_1})^\top, \ldots, ({\bm h}_{w_r})^\top$, where
 $w_1,\ldots, w_r\in \psi^{-1}(Q)$.
 Below we show that there exists a set $K\subseteq [d]\setminus \{1\} $ such that
 the matrix $\bm U$ is invertible over $\mathbb{F}$.
 Assume for contradiction that $\bm U$ is not invertible for all
 $K\subseteq [d]\setminus \{1\}$. Then
  $((\bm h_{w_1})_{[d]\setminus\{1\}},\ldots,(\bm h_{w_r})_{[d]\setminus\{1\}})^\top$
  must be  an $r\times (d-1)$ matrix of rank $<r$.
Consequently, there exist $r$ constants $\alpha_1,\ldots,\alpha_r\in \mathbb{F}$, which are not all 0,  such that
     $\sum_{i\in[r]}\alpha_i (\bm h_{w_i})^\top=(\beta,0,\ldots,0)$.
If $\beta \ne 0$, then $\{P_{\psi(w_1)},\ldots,P_{\psi(w_r)}\}$ must be an  authorized set of participants. Due to monotonicity, $Q$ is   authorized as well, which
contradicts to the   fact that $Q$ is unauthorized.
     If $\beta=0$, then the rows  $({\bm h}_{w_1})^\top, \ldots, ({\bm h}_{w_r})^\top$
     must be linearly dependent over $\mathbb{F}$, which
     contradicts to our choice of the rows.

  \section{Proof for Theorem 1}
  \label{app:thm1}

		Let $\Gamma'$ be the access structure realized by $\mathcal M'$.
		It suffices to show that $\Gamma'=\Gamma_{\cdot Q}$.

		Firstly, we show that $\Gamma'\subseteq \Gamma_{\cdot Q}$.
		For any  $A\in \Gamma'$, there   exist  a set of constants $\{\alpha_i': \psi(i)\in A\}$ such that
			$\bm t=\sum_{i\in \psi^{-1}(A)} \alpha_i' \bm{h}_i'$.
		Note that $\bm{h}_i'=\bm{h}_i-\frac{h_{ik}}{h_{nk}}\bm{h}_n$ for every $i\in[n-1]$.
		So we have that
		\begin{equation*}
			\bm t=
			\sum_{i\in \psi^{-1}(A)} \alpha_i' \left(\bm{h}_i-
			\frac{h_{ik}}{h_{nk}}\bm{h}_n\right),
		\end{equation*}
		which shows that $\bm t$ is a linear combination of
		the vectors
		$
		\{\bm h_i: \psi(i)\in A\cup Q\}.
		$
		The set $A\cup Q$ must be   authorized  in  the access structure $\Gamma$.
		Hence,  $A\in \Gamma_{\cdot Q}$.

		Secondly, we show that $ \Gamma_{\cdot Q}\subseteq \Gamma'$.
		For any $A\in \Gamma_{\cdot Q}$, we have that
		$A\cup Q\in \Gamma$.
		Then there exist a set of  constants $\{\alpha_i: \psi(i)\in A\cup Q\}$ such that
		\begin{equation}
			\label{eqn:th12}
			\bm t=\sum_{i\in\psi^{-1}(A\cup Q)} \alpha_i \bm h_i=\sum_{i\in\psi^{-1}(A)} \alpha_i \bm{h}_i+\alpha_n \bm{h}_n.
		\end{equation}
		As $Q\in 2^{\mathcal P}\setminus \Gamma,$ there exists $k  \in [d]\setminus \{1\}$
		such that $h_{nk}\neq 0$. Due to Equation \eqref{eqn:th12}, we have that
			$0=t_k=\sum_{i\in\psi^{-1}(A)} \alpha_i h_{ik}+\alpha_n h_{nk}$.
		By solving it in $\alpha_n$, we have that
		\begin{equation}
			\label{eqn:th14}
			\alpha_n=-\sum_{i\in\psi^{-1}(A)} \alpha_i \frac{h_{ik}}{h_{nk}}.
		\end{equation}
		Equations \eqref{eqn:th12} and \eqref{eqn:th14} together imply that
		
		\begin{equation*}
				\bm t=\sum_{i\in\psi^{-1}(A)} \alpha_i \left(\bm{h}_i-\frac{h_{ik}}{h_{nk}}
				\bm{h}_n\right)=\sum_{i\in \psi'^{-1}(A)} \alpha_i \bm{h}_i',
		\end{equation*}
		i.e., $\bm t^\top$ is a linear combination of the row vectors of $\bm H'$ labeled by
		$(\psi')^{-1}(A)$. Hence, $A\in \Gamma'$.

\section{Proof for Theorem 2}
\label{app:thm2}
		Let $\Gamma'$ be the access structure realized by $\mathcal M'$.
		It suffices to show that $\Gamma'=\Gamma_{\cdot Q}$.
		
		Firstly, we show that $\Gamma'\subseteq \Gamma_{\cdot Q}$.
		For any  $A\in \Gamma'$, there  exist  a set of constants $\{\alpha_i': \psi(i)\in A\}$   such that
		\begin{equation}
			\label{eqn:th21}
			\bm t^\top=\sum_{i\in \psi^{-1}(A)} \alpha_i' \bm{h}_i'^\top.
		\end{equation}
		Note that    $\bm{h}_i'^\top=\bm{h}_i^\top-(\bm h_{i}^\top)_K\cdot \bm{U}^{-1}\cdot (\bm{h}_{w_1}, \ldots, \bm{h}_{w_r})^\top$  for each $i\in[n-m]$.
		Equation (\ref{eqn:th21}) can be translated into
\begin{align*}
				\bm t^\top=\sum_{i\in \psi^{-1}(A)} \alpha_i' \left(\bm{h}_i^\top-(\bm h_i^\top)_K\cdot \bm{U}^{-1}\cdot (\bm{h}_{w_1}, \ldots, \bm{h}_{w_r})^\top\right).
\end{align*}
		Note that for every $i\in \psi^{-1}(A)$, $(\bm h_i^\top)_K\cdot \bm{U}^{-1}$ is an $r$-dimensional row vector. It is not difficult to observe that
		$
		\sum_{i\in \psi^{-1}(A)} \alpha_i'(\bm h_i^\top)_K\cdot \bm{U}^{-1}\cdot (\bm{h}_{w_1}, \ldots, \bm{h}_{w_r})^\top
		$
		is a linear combination of
		the vectors $\{\bm{h}_{w_1}^\top, \ldots, \bm{h}_{w_r}^\top\}$.
 Hence, $\bm t$ is a linear combination of
		the vectors
		$
		\{\bm h_i: \psi(i)\in A\cup \{P_{w_1},\ldots,P_{w_r}\}\}.
		$
		The set $A\cup \{P_{w_1},\ldots,P_{w_r}\}$ must be   authorized  in $\Gamma$. Due to monotonicity, we have that $A\cup Q\in \Gamma$ and thus
 $A\in \Gamma_{\cdot Q}$.

		Secondly, we show that $ \Gamma_{\cdot Q}\subseteq \Gamma'$.
		For any $A\in \Gamma_{\cdot Q}$, we have that
		$A\cup Q\in \Gamma$.
		Then there exist  a set of  constants $\{\alpha_i: \psi(i)\in A\cup Q\}$ such that
		\begin{equation}
			\label{eq21}
			\bm t^\top=\sum_{i\in\psi^{-1}(A)} \alpha_i \bm{h}_i^\top+\sum_{j\in\psi^{-1}(Q)}\alpha_j \bm{h}_j^\top.
		\end{equation}
		Since the rank of the matrix  $(\bm h_{n-m+1},\ldots, \bm h_{n})^\top$ is $r$ and
 $W$ labels a set of $r$ linearly independent rows of the matrix, there must exist a set of constants $\{\alpha'_{w}:w\in W\}$ such that
		\begin{equation}
			\label{eq22}
			\sum_{j\in\psi^{-1}(Q)}\alpha_j \bm{h}_j^\top=\sum_{w\in W}\alpha_w \bm{h}_w^\top.
		\end{equation}
		Due to  Equations \eqref{eq21} and \eqref{eq22}, we have that
		\begin{equation}
			\label{eqn:th24}
			\begin{split}
				&\sum_{i\in\psi^{-1}(A)} \alpha_i (\bm{h}_i^\top)_K+\sum_{j\in\psi^{-1}(Q)}\alpha_j (\bm{h}_j^\top)_K\\
				=&\sum_{i\in\psi^{-1}(A)} \alpha_i (\bm{h}_i^\top)_K+\sum_{w\in W}\alpha_w \bm{h}_w^\top\\
				=&\sum_{i\in\psi^{-1}(A)} \alpha_i (\bm{h}_i^\top)_K+(\alpha'_{w})_{w\in W}\cdot \bm U\\
				=&\hspace{3.9mm}(\bm t^\top)_K\\
=&\hspace{6mm}\bm 0.
			\end{split}
		\end{equation}
		By solving the linear equation system \eqref{eqn:th24}, we have that
		\begin{equation}
			\label{eq23}
			(\alpha'_{w})_{w\in W}=-\sum_{i\in\psi^{-1}(A)} \alpha_i (\bm{h}_i^\top)_K\cdot \bm U^{-1}.
		\end{equation}
		Equations \eqref{eq21}, \eqref{eq22} and \eqref{eq23} together imply that
		
		\begin{equation*}
			\begin{split}
				\bm t^\top
				=&\hspace{1.5mm}\sum_{i\in\psi^{-1}(A)} \alpha_i \bm{h}_i^\top+\sum_{w\in W} \alpha'_{w} \bm{h}_{w}^\top\\
				=&\hspace{1.5mm}\sum_{i\in\psi^{-1}(A)} \alpha_i \bm{h}_i^\top+(\alpha'_{w})_{w\in W}\cdot (\bm{h}_{w_1}, \ldots, \bm{h}_{w_r})^\top\\
			 				=&\hspace{1.5mm}\sum_{i\in\psi^{-1}(A)} \alpha_i \left(\bm{h}_i^\top-(\bm h_{i}^\top)_K\cdot \bm{U}^{-1}\cdot (\bm{h}_{w_1}, \ldots, \bm{h}_{w_r})^\top\right)\\
				=&\sum_{i\in(\psi')^{-1}(A)} \alpha_i \bm{h}_i'^\top.
			\end{split}
		\end{equation*}
Thus, $\bm t^\top$ is a linear combination of the rows of $\bm H'$ labeled by
		$(\psi')^{-1}(A)$. Therefore, $A\in \Gamma'$.

\section{A Toy Example of Storage Relocation}
	\label{app_toyexmaple}
	Let $\mathcal P=\{P_1,P_2,P_3,P_4\}$ be a set of four servers. Consider an access structure $\Gamma$ over $\mathcal P$ such that  $\Gamma^-=\{\{P_1,P_2,P_4\},\{P_1,P_3,P_4\}\}$. Let $\mathcal M=(\mathbb F_2,\bm H,\bm t,\psi)$ be an MSP and   an ideal LSSS for $\Gamma$, where
	\[\bm{H}=(\bm h_1,\bm h_2,\bm h_3,\bm h_4)^\top=
	\begin{bmatrix}
		1	&	0	&	1	\\
		0	&	1	&	1	\\
		0	&	1	&	1	\\
		0	&	1	&	0	\\	
	\end{bmatrix},
	\]
	$\bm t=(1,0,0)^\top$, and $\psi(i)=P_i$ for each $i\in[4]$. To share a secret $s\in \mathbb F_2$ with $\mathcal M$, the dealer chooses  $r_2,r_3\leftarrow \mathbb F_2$, defines $\bm v=(s,r_2,r_3)^\top$, and gives $s_i=\bm h_i^\top\bm v$ to $P_i$ for all $i\in[4]$. Let $Q=\{P_4\}$ be an unauthorized subset. We  show how to use the methods in Section {\ref{sec:appmc}} to solve (p1).
	
\vspace{1mm}
	\noindent \textbf{Our method.}
 By   running {Algorithm 1} on input $({\cal M},Q)$ with
	$k=2$, we get an LSSS $\mathcal M'=(\mathbb F,\bm H',\bm t,\psi')$
for $\Gamma_{\cdot Q}$,   where
	\[\bm{H}'=(\bm h_1',\bm h_2',\bm h_3')^\top=
	\begin{bmatrix}
		1	&	0	&	1	\\
		0	&	0	&	1	\\
		0	&	0	&	1	\\
	\end{bmatrix}
	\]
	and $\psi'(i)=P_{i}$ for all $i\in [3]$. Our method requires
 each server  $P_i\in \mathcal{P}\setminus Q$ to compute and store a new share
			$s_i'=s_i-\frac{h_{i2}}{h_{42}}s_4$.
	Specifically,  $s_1'=s_1$,  $s_2'=s_2-s_4$, and  $s_3'=s_3-s_4$.

	\vspace{1mm}
	\noindent \textbf{Martin's method.}
To use Martin's method, the SSS $\cal M$ can be represented as the following matrix
	\[\bm M=\bordermatrix{
		&	P_1	&	P_2	&	P_3	&	P_4	\cr
		(0,(0,0))	&	0	&	0	&	0	&	0	\cr
		(0,(0,1))	&	1	&	1	&	1	&	0	\cr
		(0,(1,0))	&	0	&	1	&	1	&	1	\cr
		(0,(1,1))	&	1	&	0	&	0	&	1	\cr
		(1,(0,0))	&	1	&	0	&	0	&	0	\cr
		(1,(0,1))	&	0	&	1	&	1	&	0	\cr
		(1,(1,0))	&	1	&	1	&	1	&	1	\cr
		(1,(1,1))	&	0	&	0	&	0	&	1	\cr
	},
	\]
	whose rows and columns   are labeled by the elements of $\mathbb F_2\times\mathbb F_2^2$ and $\mathcal P$, respectively. By choosing $\bm \alpha=(0)$,   the
SSS   can be represented with
	\[\bm M'=\bordermatrix{
		&	P_1	&	P_2	&	P_3	\cr
		(0,(0,0))	&	0	&	0	&	0	\cr
		(0,(0,1))	&	1	&	1	&	1	\cr
		(1,(0,0))	&	1	&	0	&	0	\cr
		(1,(0,1))	&	0	&	1	&	1	\cr
	}.
	\]
It suffices to keep the shares of  $ P_1,P_2$ and $ P_3 $ unchanged and
send  the share of $Q=\{P_4\}$   to a public storage.
	
	\vspace{1mm}
	\noindent \textbf{Nikov-Nikova method.}
This method will   construct an MSP $\mathcal M'=(\mathbb F_2,\bm H',\bm t,\psi')$  for $\Gamma_{\cdot Q}$, where
	\[\bm{H}'=
	\begin{bmatrix}
		1	&	0	&	1	\\
		0	&	1	&	1	\\
		0	&	1	&	1	\\
		0	&	1	&	0	\\
		0	&	1	&	0	\\
		0	&	1	&	0	\\
	\end{bmatrix}
	\]
	and $(\psi')^{-1}(P_{i})=\{i,i+3\}$ for every $i\in [3]$.
	It requires  $P_4$  to transfer its share to every server in
$\{P_1,P_2,P_3\}$ and requires each of these servers
to additionally store   $s_4$. The new scheme is  \emph{non-ideal}.
	
	\vspace{1mm}
	\noindent \textbf{Extended Nikov-Nikova method.}
This method requires every
 authorized subset to keep a copy of the share $s_4$ of $Q$. As
  $(\Gamma_{\cdot Q})^-=\{\{P_1,P_2\},\{P_1,P_3\}\}$,
  the method can be optimized by $P_1$ storing  $s_4$.

\section{A Solution to Storage Relocation} \label{appendixE}
Here we give a more intuitive solution for the storage relocation problem, which {is equivalent} to our method in Section {\ref{ourmd}}.

Let $\mathcal M=(\mathbb F,\bm H,\bm t,\psi)$ be an MSP and an ideal LSSS for $\Gamma$, where $\bm H=(\bm h_1,\ldots, \bm h_n)^\top\in \mathbb F^{n \times d}$,
	$\bm t=(1,0,\ldots,0)^\top\in \mathbb{F}^d$, and $\psi(i)=P_i$ for every $i\in[n]$.
	For any unauthorized subset $Q\subseteq \mathcal P$,
the shares $\{s_j\}_{P_j\in Q}$ leaves  the secret $s$ completely undetermined, i.e.,
from $Q$'s view of point, any element of $\mathbb{F}$ could be $s$.
In particular, there must exist a vector $\bm v'=(0,r'_2,\ldots,r'_d)^\top\in \mathbb F^d$ such that  $s_j=\bm h_j^\top\bm v'$ for all $P_j\in Q$.
Such a vector ${\bm v}'$ may be not uniquely determined.
Based on these observations,  an alternative solution of   (p1) can be described as follows:
	\emph{
		\begin{itemize}
			\item Every server   $P_j\in Q$ broadcasts its shares $s_j$.
			\item The servers in ${\mathcal P}\setminus Q$ agree on  a vector
  $\bm v'$ such that $s_j=\bm h_j^\top   \bm v'$ for all $P_j\in Q$.
  \item Every server $P_i\in {\cal P}\setminus Q$ replaces its share $s_i$ with
$s'_i=s_i-\bm h_i^\top\bm v'$.
	\end{itemize}}
\noindent	
Note that there is a vector ${\bm v}=(s,r_2,\ldots,r_d)$ such that
 $s_i={\bm h}^\top_i  {\bm v}$ for all $i\in[n]$.
The correctness of this solution follows from the fact  that $\{s'_i\}_{i\in[n]}$ is a set of  valid shares of   $s$, which are generated by choosing a random vector
$\bm v''=\bm v-\bm v'=(s,r_2-r'_2,\ldots,r_d-r_d')$ and
 setting $s'_i={\bm h}^\top_i  {\bm v}''$, and the  fact that
 $s'_j=0$ for  all $P_j\in Q$ and   the
 shares of $Q$ are actually not needed in any reconstruction.


Now we show that this method is equivalent to our method.
Due to Lemma 1, if we let ${\bm H_Q}=\big((\bm h_i)_{\psi(i)\in Q}\big)^\top$ and assume that ${\rm rank}({\bm H}_Q)=r$, then there  exists a set $W=\{w_1,\ldots,w_r\}\subseteq \psi^{-1}(Q)$ and a set  $K=\{k_1,\ldots,k_r\}\subseteq [d]\setminus\{1\}$ such that the order-$r$ square matrix $\bm U=((\bm h_{w_1})_K,\ldots,(\bm h_{w_r})_K)^\top$ is invertible over $\mathbb{F}$. To find a vector $\bm v'=(0,r'_2,\ldots,r'_d)^\top\in \mathbb F^d$ such that  $s_j=\bm h_j^\top\bm v'$ for all $P_j\in Q$, it suffices to solve the equation $(\bm h_{w_1},\ldots,\bm h_{w_r})^\top \bm v'=(s_{w_1},\ldots,s_{w_r})^\top$ in $\bm v'$. It is not difficult to observe that there exists a vector $\bm v'$, where $r'_i=0$ for each $i\in ([d]\setminus\{1\})\setminus K$ and $(r'_{k_1},\ldots,r'_{k_r})^\top=\bm U^{-1}(s_{w_1},\ldots,s_{w_r})^\top$, such that $s_j=\bm h_j^\top\bm v'$ for all $P_j\in Q$. Then for every server $P_i\in \mathcal P\setminus Q$,
\begin{equation*}
	\begin{split}
	s'_i=s_i-\bm h_i^\top \bm v'&=s_i-\sum_{k\in K} h_{ik}r'_k\\
	&=s_i-(\bm h_i^\top)_K\cdot (r'_{k_1},\ldots,r'_{k_r})^\top\\
	&=s_i-(\bm h_i^\top)_K\cdot \bm U^{-1}(s_{w_1},\ldots,s_{w_r})^\top.
	\end{split}
\end{equation*}
This shows that the two methods are equivalent. Although our method is less intuitive, it can be well applied to construct the CP-ABE-CAS, the  above method cannot.

\section{Bilinear Pairing}
\label{app:bp}

\begin{definition}[{\bf Bilinear Pairing}]
	Let $\mathbb G$ and $\mathbb G_T$ be two multiplicative cyclic groups of prime order $p$. Let $g$ be a generator of $\mathbb G$ and $e:\mathbb G\times\mathbb G\rightarrow \mathbb G_T$ be a bilinear map such that: (1) $e(u^a,v^b)=e(u,v)^{ab}$ for all $u,v\in \mathbb G$ and $a,b\in \mathbb Z_p$; (2) $e(g,g)\ne 1$.
	We say that $\mathbb G$ is a bilinear group if the group operations in $\mathbb G$ and $\mathbb G_T$ as well as the bilinear map $e$ are efficiently computable.
 \end{definition}

\section{Security of our CP-ABE-CAS scheme} \label{appendix:CPABEsecur}
	    \begin{theorem}
	    \label{3}
	        A CP-ABE-CAS is secure if Waters' CP-ABE scheme is secure.
	    \end{theorem}
	    \begin{proof}
	        Assume that a CP-ABE-CAS is not fully secure. Then there exists an adversary $\mathcal A$ who has a non-negligible advantage in $G_2$.
	        Next we show that an adversary $\mathcal B$ can win the game $G_1$ by means of the advantage of $\mathcal A$.
	        \begin{itemize}
	            \item \textit{Setup:} The challenger runs the setup algorithm and gives the public parameter $\mathrm{PK}$ to $\mathcal A$ and $\mathcal B$.
	            \item \textit{Challenge:} $\mathcal A$ chooses two equal length messages $M_0, M_1$, an ideal access structure $\Gamma^*$, an LSSS $\mathcal M$ for $\Gamma^*$, and an unauthorized subset $Q=\{y\}$ of $\Gamma^*$, (w.l.o.g, assume that $\psi^{-1}(y)=\ell$) and then sends them to $\mathcal B$. $\mathcal B$ runs Algorithm 1 on input $(\mathcal M,Q)$, obtains the output $\mathcal M'$ for $\Gamma^*_{\cdot Q}$ and records $k$. Then $\mathcal B$ sends $M_0,M_1,\Gamma^*_{\cdot Q},\mathcal M'$ to the challenger. The challenger randomly chooses $b\in\{0,1\}$, and encrypts $M_b$ under $\Gamma^*_{\cdot Q}$, producing $\mathrm{CT}^*_{\cdot Q}=\{\mathcal M',C,C',(C_i,D_i)_{i\in[\ell-1]}\}$. It gives $\mathrm{CT}^*_{\cdot Q}$ to $\mathcal B$.
	
	            \item \textit{Adversary $\mathcal B$:} The adversary $\mathcal B$ chooses two random values $r_\ell,s_\ell$ and creates
				\begin{align*}
					\mathrm{CK}_Q&=r_\ell, \hspace{4em} C_i'=C_i\cdot g^{as_\ell\cdot \frac{h_{ik}}{h_{\ell k}}},\\
					C_{\ell}'&=g^{as_\ell} T_y^{-r_{\ell}},\hspace{0.7em} D_{\ell}=g^{r_\ell}\\
	            \end{align*}
	            for each $i\in [\ell]\setminus\{\psi^{-1}(y)\}$. It then produces a new ciphertext $\mathrm{CT}^*=\{\mathcal M,C,C',(C_i',D_i)_{i\in[\ell]}\}$. It gives $\mathrm{CT}^*,\mathrm{CT}^*_{\cdot Q},\mathrm{CK}_Q$ to the adversary $\mathcal B$.
	            \item \textit{Adversary $\mathcal A$:} The adversary $\mathcal A$ returns a guess $b'$.
	            \item \textit{Guess:} The adversary $\mathcal B$ output $b'$.
	 \end{itemize}
	
	 Since $\mathcal A$ has a non-negligible advantage in $G_2$, $\mathcal B$ will also have a non-negligible advantage to win $G_1$.
	 \end{proof}
	
	 Waters' CP-ABE scheme has been showed to be secure under the decisional $q$-parallel Bilinear Diffie-Hellman Exponent Assumption with $q\ge\ell,d$, where $\ell\times d$ is the size of the matrix $\bm H$ in the LSSS $\mathcal M$, so the CP-ABE-CAS is also secure under the same assumption according to Theorem \ref{3}.

\section{On the Data Integrity}
\label{DI}

        The data integrity is another critical security concern in cloud storage and requires that after the access structure contraction/policy update, the original data can still be recovered from the updated shares/ciphertexts.

      \vspace{1mm}
      \noindent
{\bf Single-cloud storage.}
  In single-cloud storage,  the data integrity problem has been considered in \cite{WZMZ23,GSB+21}. Ge et al. \cite{GSB+21} proposed a revocable CP-ABE scheme that solves the
   data integrity problem  by adding a commitment to  Waters' scheme \cite{Wat11}.
 In our model, we assume   the server is honest-but-curious and  only the data owner {is allowed} to contract the policy. If we consider   {\em malicious} servers/users,
 two threats may appear:
         \begin{itemize}
             \item A malicious server may tamper with the ciphertexts.
             \item A malicious user may impersonate the data owner and sends an arbitrary contraction key $\widehat{\rm CK}$ to the server.
         \end{itemize}
         Fortunately, both the above two threats can be relieved by applying our policy updating techniques from Section \ref{sec:app} to the scheme of  \cite{GSB+21}.

      \vspace{1mm}
      \noindent
{\bf Multi-cloud storage.}
        In multi-cloud storage, provable data possession (PDP) \cite{ABC+07,HCY+21}  and proofs of retrievability (PoR) \cite{SW08,ADJ+21} have been invented for ensuring the data integrity. PDPs ensure data integrity through the validation of file block signatures.
        PoRs incorporate error-correcting codes to provide not only the data integrity but also the recoverability of corrupted data.   In addition, techniques like homomorphic authenticators (HAs) \cite{CF13,FXZ22,ZS14}   provide an alternative approach to solve the data integrity problem  in multi-cloud storage through authenticating the shares on the servers.

\vfill

\end{document}